\documentclass[journal,comsoc]{IEEEtran}
\usepackage[T1]{fontenc}% optional T1 font encoding

\ifCLASSINFOpdf
\else
\fi
\usepackage{cite}
\usepackage{amsmath,amssymb,amsfonts}
\usepackage{algorithmic}
\usepackage{graphicx}
\usepackage{textcomp}
\usepackage{amsfonts}
\usepackage{amsmath}
\usepackage{graphicx}
\usepackage{color,soul}
\usepackage{environ}
\usepackage{upquote}
\usepackage{algorithm}
\usepackage{color}
\usepackage{algorithm,algorithmic}
\usepackage{enumitem}
\usepackage{array}

\usepackage{amssymb}

\usepackage{mathtools}
\usepackage{upgreek}
\usepackage{caption}
\usepackage{subcaption}
\usepackage{url}
\usepackage{cite}
\usepackage{amsmath}
\usepackage{eufrak}
\usepackage[autostyle]{csquotes}
\usepackage{verbatim}
\usepackage{amsmath,amssymb,amsthm}
\usepackage{tabularx}

\usepackage{multirow,array}
\usepackage[utf8]{inputenc}
\usepackage[english]{babel}

\newtheorem{definition}{Definition}
\newtheorem{proposition}{Proposition}

\usepackage[cmintegrals]{newtxmath}
\begin{document}
\title{Controlling the Outbreak of COVID-19: A Noncooperative Game Perspective
}

\author{Anupam Kumar Bairagi, \IEEEmembership{Member, IEEE,}
	Mehedi Masud,~\IEEEmembership{Senior Member,~IEEE,}
	Do Hyeon Kim,~\IEEEmembership{}
	Md. Shirajum Munir,~\IEEEmembership{Student Member,~IEEE,}
	Abdullah Al Nahid,~\IEEEmembership{}
	Sarder Fakhrul Abedin,~\IEEEmembership{Student Member,~IEEE,}
	Kazi Masudul Alam,~\IEEEmembership{}
	Sujit Biswas,~\IEEEmembership{Member,~IEEE,}
	Sultan S~Alshamrani,~\IEEEmembership{}
	Zhu Han,~\IEEEmembership{Fellow,~IEEE,}
	and~Choong~Seon~Hong,~\IEEEmembership{Senior~Member,~IEEE}% <-this % stops a space
	\thanks{Anupam Kumar Bairagi, Do Hyeon Kim, Md. Shirajum Munir, Sarder Fakhrul Abedin, and Choong~Seon~Hong  are with the Department of Computer Science and Engineering, Kyung Hee University, Yongin-si 17104, Republic of Korea.}% <-this % stops a space
	\thanks{Mehedi Masud is with the Department of Computer Science, Taif University, Taif, KSA.}
	\thanks{Mehedi Masud and Sultan S~Alshamrani is with the Department of Computer Science, Taif University, Taif, KSA.}
	\thanks{Abdullah Al Nahid is with the Electronics and Communication Engineering Discipline, Khulna University, Bangladesh.}
	\thanks{Kazi Masudul Alam is with the Computer Science and Engineering Discipline, Khulna University, Bangladesh.}
	\thanks{Sujit Biswasis is with the Department of Computer Science and Engineering, Faridpur Engineering College, Bangladesh.}
	\thanks{Sultan S~Alshamrani is with theDepartment of Information Technology, Taif University, Taif, KSA.}
	\thanks{Zhu Han is with the Electrical and Computer Engineering Department, University of Houston, Houston, TX 77004, USA, and also with the Department of Computer Science and Engineering, Kyung Hee University, Yongin-si 17104, Republic of Korea.}
	\thanks{Corresponding author: Choong Seon Hong (e-mail: cshong@khu.ac.kr)}
	\thanks{This research was supported by the MSIT(Ministry of Science and ICT), Korea, under the Grand Information Technology Research Center support program(IITP-2020-2015-0-00742) supervised by the IITP(Institute for Information \& communications Technology Planning \& Evaluation). Dr. CS Hong is the corresponding author.}
	\thanks{© 2020 IEEE. Personal use of this material is permitted.  Permission from IEEE must be obtained for all other uses, in any current or future media, including reprinting/republishing this material for advertising or promotional purposes, creating new collective works, for resale or redistribution to servers or lists, or reuse of any copyrighted component of this work in other works.}}

%\corresp{Corresponding author: Choong Seon~Hong (e-mail: cshong@khu.ac.kr).}

\markboth{Accepted article by IEEE Access. DOI: 10.1109/ACCESS.2020.3040821}%
{Shell \MakeLowercase{\textit{et al.}}: Bare Demo of IEEEtran.cls for IEEE Communications Society Journals}



\maketitle

\begin{abstract}
\emph{COVID-19} is a global epidemic. Till now, there is no remedy for this epidemic. However, isolation and social distancing are seemed to be effective preventive measures to control this pandemic. Therefore, in this paper, an optimization problem is formulated that accommodates both isolation and social distancing features of the individuals. To promote social distancing, we solve the formulated problem by applying a noncooperative game that can provide an incentive for maintaining social distancing to prevent the spread of COVID-19. Furthermore, the sustainability of the lockdown policy is interpreted with the help of our proposed game-theoretic incentive model for maintaining social distancing where there exists a Nash equilibrium. Finally, we perform an extensive numerical analysis that shows the effectiveness of the proposed approach in terms of achieving the desired social-distancing to prevent the outbreak of the COVID-19 in a noncooperative environment. Numerical results show that the individual incentive increases more than $85\%$ with an increasing percentage of home isolation from $25\%$ to $100\%$ for all considered scenarios. The numerical results also demonstrate that in a particular percentage of home isolation, the individual incentive decreases with an increasing number of individuals.
\end{abstract}


\IEEEpeerreviewmaketitle

\section{Introduction}
\label{sec:introduction}
\IEEEPARstart{T}{he} novel Coronavirus (2019-nCoV or COVID-19) is considered to be one of the most dangerous pandemics of this century. COVID-19 has already affected every aspect of individual's life i.e. politics, sovereignty, economy, education, religion, entertainment, sports, tourism, transportation, and manufacturing. It was first identified in Wuhan City,  China on December 29, 2019, and within a short span of time, it spread out worldwide \cite{AWu2019,CHuang2020}. The World Health Organization (WHO) has announced the COVID-19 outbreak as a Public Health Emergency of International Concern (PHEIC) and identified it as an epidemic on January 30, 2020 \cite{WHO2020}. Till July 23, 2020, COVID-19 has affected $215$ countries and territories throughout the globe and 2 international conveyances \cite{Worldmeter2020}.

The recent statistics on COVID-19 also indicate that more than $15,371,829$ persons have been affected in different ways \cite{Worldmeter2020,JHUM2020}. Currently, the ten most infected countries are USA, Brazil, India, Russia, South Africa, Peru, Mexico, Chile, Spain, UK, \textcolor{black}{and these countries contributed more than $68\%$ of worldwide cases}. Since the outbreak, the total number of human death and recovery to/from COVID-19 are $630,138$ and $9,348,761$, respectively \cite{Worldmeter2020,JHUM2020} (till July 23, 2020). The fatality of human life due to COVID-19 is frightening in numerous countries. For instance, among the highest mortality rates countries, $70\%$ of the mortality belongs to the top $8$ countries due to COVID-19. Furthermore, the percentages of affected cases for male and female are around $55.21\%$ and $44.79\%$, whereas these values are about $76\%$ and $24\%$, respectively in death cases globally \cite{Global5050}. Different countries are undertaking different initiatives to reduce the impact of the COVID-19 epidemic, but there is no clear-cut solution to date.

One of the most crucial tasks that countries need to do for understanding and controlling the spread of COVID-19 is testing. Testing allows infected bodies to acknowledge that they are already affected. This can be helpful for taking care of them, and also to decrease the possibility of contaminating others. In addition, testing is also essential for a proper response to the pandemic. It allows carrying evidence-based steps to slow down the spread of COVID-19. However, to date, the testing capability for COVID-19 is quite inadequate in most countries around the world.  South Korea was the second COVID-19 infectious country after China during February 2020. However, mass testing may be one of the reasons why it succeeded to diminish the number of new infections in the first wave of the outbreak since it facilitates a rapid identification of potential outbreaks \cite{Balilla2020}. For detecting COVID-19, two kinds of tests are clinically carried out: (i) detection of virus particles in swabs collected from the mouth or nose, and (ii) estimating the antibody response to the virus in blood serum.       

This COVID-19 epidemic is still uncontrolled in most countries. As a result, day by day, the infected cases and death graph are rising exponentially. However, researchers are also focusing on the learning-based mechanism for detecting COVID-19 infections \cite{Maghdid2020,Rao2020,LLi2020,XXu2020,YWang2020,Gozes2020,CZheng2020}. This approach can be cost-effective and also possibly will take less time to perform the test. \textcolor{black}{Some other studies \cite{GQian2020,PWang2020,NChen2020,WGuan2020,ZHu2020,HChen2020,FHJAI2020,Jiao2020, Din2020} focus on finding the spreading behavior of COVID-19 by using known epidemic models like Susceptible - Infectious - Recovered (SIR), Susceptible - Infectious - Recovered - Susceptible (SIRS), Susceptible - Exposed - Infectious - Recovered (SEIR), Susceptible-Infected-Hibernator-Removed (SIHR), Susceptibl- Infected-Diagnosed-Ailing-Recognized-Threatened-Healed-Extinct (SIDARTHE) etc. However, all of these epidemic models are confined by the hypothesis of constant recovery rates, and they also struggle to reveal the system dynamics when there is a limited coupling between subpopulations \cite{Roberts2015}. Moreover, These models are inadequate to capture typically stochastic aspects, like fade-out, extinction, and lack of synchrony due to arbitrary delays \cite{Rock2014}. Besides, the limitations of these models, they may present the epidemic scenario but have little impact on reducing or controlling the causes.} However, the infected cases of the COVID-19 can be reduced by maintaining a certain social distance among the individuals. In particular, to maintain such social distancing, self-isolation, and community lockdown can be possible approaches. Thus, it is imperative to develop a model so that the social community can take a certain decision for self-isolation/lockdown to prevent the spread of COVID-19.

To the best of our knowledge, there is no study that focuses on the mathematical model for monitoring and controlling individual in a community setting to prevent this COVID-19 epidemic. Thus, the main contribution of this paper is to develop an effective mathematical model with the help of global positioning system (GPS) information to fight against COVID-19 epidemic by monitoring and controlling individual. To this end, we make the following key contributions:
\begin{itemize}  
	%\item First, we analyze the real-world dataset to realize the worldwide severity of COVID-19 epidemic and also show the predicted results for infected and active cases of COVID-19.
	\item First, we formulate an optimization problem for maximizing the social utility of individual considering both isolation and social distancing. Here, the optimization parameters are the positions of individual.  
	\item Second, we reformulate the objective function which is incorporated with the social distancing feature of an individual as a noncooperative game. Here, we show that home isolation is the dominant strategy for all the individuals (players) of the game. We also prove that the game has a Nash Equilibrium (NE).
	\item Third, we interpret the sustainability of lockdown policy with the help of our model. 
	\item Finally, we evaluate the effectiveness of the proposed approach with the help of extensive numerical analysis.   
\end{itemize}

The remainder of the paper is organized as follows. In Section \ref{Lit_Rev}, we present the literature review. We explain the system model and present the problem formulation in Section \ref{Sys_model}. The proposed solution approach of the above-mentioned problem is addressed in Section \ref{Sol_Game}. We interpret the sustainability of lockdown policy with our model in Section \ref{Sustainability}. In Section \ref{Per_eva}, we provide numerical analysis for the proposed approach. \textcolor{black}{We present the limitation of the current study in Section \ref{Limitation}.} Finally, we draw some conclusions in Section \ref{Conc}. 

\section{Literature Review}
\label{Lit_Rev}
COVID-19 is the seventh coronavirus identified to contaminate humans. Individuals were first affected by the 2019-nCoV virus from bats and other animals that were sold at the seafood market in Wuhan \cite{RLu2020,KJalava2020}. Afterward, it began to spread from human to human mainly through respiratory droplets produced while individual sneeze cough or exhaling \cite{WHO2020}.  

\textcolor{black}{In \cite{GQian2020}, the authors present a generalized fractional-order SEIR model (SEIQRP) for predicting the potential outbreak of contagious diseases alike COVID-19.  They also present a modified SEIQRP model in their work, namely the SEIQRPD model.  With the real data of COVID-19, they have shown that the proposed model has a more reliable prediction capability for the succeeding two weeks. In \cite{PWang2020}, the authors introduce a Bayesian Heterogeneity Learning approach for Susceptible-Infected-Removal-Susceptible (SIRS) model. They formulate the SIRS model into a hierarchical structure and assign the Mixture of Finite mixtures priors for heterogeneity learning. They utilize the methodology to investigate the state level COVID-19 data in the U.S.A. The authors induce an innovative neurodynamical model of epidemics called Neuro-SIR in \cite{NChen2020}. The proposed approach allows the modeling of pandemic processes in profoundly different populations and contagiousness contexts. In \cite{WGuan2020}, the authors propose a mobility-based SIR model for epidemics considering the pandemic situations like COVID-19. The proposed model considers the population distribution and connectivity of different geographic locations across the globe. The authors propose a noble mathematical model for presenting the COVID-19 pandemic by fractional-order SIDARTHE model in \cite{ZHu2020}. They prove the existence of a steady solution of the fractional-order COVID-19 SIDARTHE model. They also produce the necessary conditions for the fractional order of four proposed control strategies. In \cite{HChen2020}, the authors propose a conceptual mathematical model for the epidemic dynamics using four compartments, namely Susceptible, Infected, Hospitalized, and Recovered. They investigate the stability of the equilibrium for the model using the basic reproduction number for knowing the austerity. In \cite{FHJAI2020}, the authors develop a pandemic model to inquire about the transmission dynamics of the COVID-19.  Here, they assess the theoretical impact of probable control invasions like home quarantine, social distancing, cautious behavior, and other self-imposed measures. They apply the Bayesian approach and authorized data to figure out some of the model parameters.  In \cite{Jiao2020}, the authors introduce a SIHR model to prognosticate the course of the epidemic for finding an effective control scheme. The model parameters are estimated based on fitting to the published data of Hubei province, China. In \cite{Din2020}, the authors present a mathematical model for COVID-19 based on three different compartments, namely susceptible, infected, and recovered classes.  They also present some qualitative viewpoints for the model, i.e., the existence of equilibrium and its stability issues.}

%\normalsize
%	\centering
%	\caption{List of Abbreviations}\label{Table_Epi_Cli}
%	\renewcommand{\arraystretch}{1}
%	\begin{tabular}{|c|c|c|c|c|c|} \hline
%		\textbf{Study} & \textbf{No. of Cases}&\textbf{Main Findings} & \textbf{Average Age}& \textbf{Gender(F:M)} & \textbf{Fatality Rate}\\ \hline 
%		\cite{GQian2020} & 91 & Fever(71.43\%), Cough(60.44\%), Fatigue(43.96\%) & 50 &-- &--\\ \hline
%		\cite{PWang2020} &  & Fever(71.43\%), Cough(60.44\%), Fatigue(43.96\%) & 50 &-- &--\\ \hline
%		\cite{NChen2020} &  & Fever(71.43\%), Cough(60.44\%), Fatigue(43.96\%) & 50 &-- &--\\ \hline
%		\cite{WGuan2020} & 1099 & Fever(87.90\%), Cough(67.70\%) & 47 &41.90\%:58.10\%, &1.36\%\\ \hline
%		\cite{ZHu2020} & 24 & Fever(20.8\%), Cough(20.8\%), Fatigue(20.8\%) & -- &-- &--\\ \hline
%		\cite{HChen2020} & 9 & Fever(44.44\%), Cough(44.44\%), Myalgia(33.33\%) & -- &100\%:0\% &--\\ \hline
%		\cite{JCao2020} & 18 & -- & -- &-- &--\\ \hline
%	\end{tabular}
%	\renewcommand{\arraystretch}{1}
%	$\vspace{-.4cm}$
%\end{table}

Machine learning can play an important role to detect COVID-19 infected individual based on the observatory data. The work in \cite{Maghdid2020} proposes an algorithm to investigate the readings from the smartphone’s sensors to find the COVID 19 symptoms of a patient. Some commons symptoms of COVID-19 victims like fever, fatigue, headache, nausea, dry cough, lung CT imaging features, and shortness of breath can be captured by using the smartphone. This detection approach for COVID-19 is faster than the clinical diagnosis methods.
The authors in \cite{Rao2020} propose an artificial intelligence (AI) framework for obtaining the travel history of individual using a phone-based survey to classify them as no-risk, minimal-risk, moderate-risk, and high-risk of being affected with COVID-19. The model needs to be trained with the COVID-19 infected information of the areas where s/he visited to accurately predict the risk level of COVID-19.
In \cite{LLi2020}, the authors develop a deep learning-based method (COVNet) to identify COVID -19 from the volumetric chest CT image. For measuring the accuracy of their system, they utilize community-acquired pneumonia (CAP) and other non-pneumonia CT images. The authors in \cite{XXu2020} also use deep learning techniques for distinguishing COVID-19 pneumonia from Influenza-A viral pneumonia and healthy cases based on the pulmonary CT images. They use a location-attention classification model to categorize the images into the above three groups.
Depth cameras and deep learning are applied to recognize unusual respiratory pattern of personnel remotely and accurately in \cite{YWang2020}. They propose a novel and effective respiratory simulation model based on the characteristics of original respiratory signals. This model intends to fill the gap between large training datasets and infrequent real-world data.
Multiple retrospective experiments were demonstrated to examine the performance of the system in the detection of speculated COVID-19 thoracic CT characteristics in \cite{Gozes2020}. A 3D volume review, namely  “Corona score” is employed to assess the evolution of the disease in each victim over time.
In \cite{CZheng2020}, the authors use a pre-trained UNet to fragment the lung region for automatic detection of COVID-19 from a chest CT image. Afterward, they use a 3D deep neural network to estimate the probability of COVID-19 infections over the segmented 3D lung region. Their algorithm uses 499 CT volumes as a training dataset and 131 CT volumes as a test dataset and achieves 0.959 ROC AUC and 0.976 PR AUC.
The study in \cite{Rahman2020} presents evidence of the diversity of human coronavirus, the rapid evolution of COVID-19, and their clinical and Epidemiological characteristics. The authors also develop a deep learning model for identifying COVID-19. and trained the model using a small CT image datasets. They find an accuracy of around $90\%$ using a small CT image dataset. 

In \cite{Hellewell2020},  the authors propose a stochastic transmission model for capturing the phenomenon of the COVID-19 outbreak by applying a new model to quantify the effectiveness of association tracing and isolation of cases at controlling a severe acute respiratory syndrome coronavirus 2 (SARS-CoV-2)-like pathogen. In their model, they analyze synopses with a varying number of initial cases, the basic reproduction number, the delay from symptom onset to isolation, the probability that contacts were traced, the proportion of transmission that occurred before symptom start, and the proportion of subclinical infections. They find that contact tracing and case isolation are capable enough to restrain a new outbreak of COVID-19 within 3 months. In \cite{Munir2020}, the authors present a risk-sensitive social distance recommendation system to ensure private safety from COVID-19. They formulate a social distance recommendation problem by characterizing Conditional Value-at-Risk (CVaR) for a personal area network (PAN) via Bluetooth beacon. They solve the formulated problem by proposing a two phases algorithm based on a linear normal model. In \cite{Chamola2020}, the authors mainly dissect the various technological interventions made in the direction of COVID-19 impact management. Primarily, they focus on the use of emerging technologies such as Internet of Things (IoT), drones, artificial intelligence (AI), blockchain, and 5G in mitigating the impact of the COVID-19 pandemic.

\textcolor{black}{Moreover, noncooperative game theory is used by different authors for solving resource allocation problems in communication \cite{Zhou2015, Zhou2016}. In \cite{Zhou2015}, the authors represent the resource allocation problem as a noncooperative game, where every player desires to maximize its energy efficiency (EE). In \cite{Zhou2016}, the authors modeled the distributed resource allocation problem as a noncooperative game in which every player optimizes its EE individually with the support of distributed remote radio heads.} 

The works \cite{GQian2020,PWang2020,NChen2020,WGuan2020,ZHu2020,HChen2020,Maghdid2020,Rao2020,LLi2020,XXu2020,YWang2020,Gozes2020,CZheng2020, Rahman2020, Hellewell2020, Munir2020, Chamola2020} focused on COVID-19 detection and analyzed the characteristic of its respiratory pattern. Hence, the literature has achieved a significant result in terms of post responses. In fact, it is also imperative to control the epidemic of COVID-19 by maintaining social distance. Therefore, different from the existing literature, we focus on the design of a model that can measure individual's isolation and social distance to prevent the epidemic of COVID-19. The model considers both isolation and social distancing features of individuals to control the outbreak of COVID-19.

\begin{figure}
	\centering
	\includegraphics[width=0.5\textwidth]{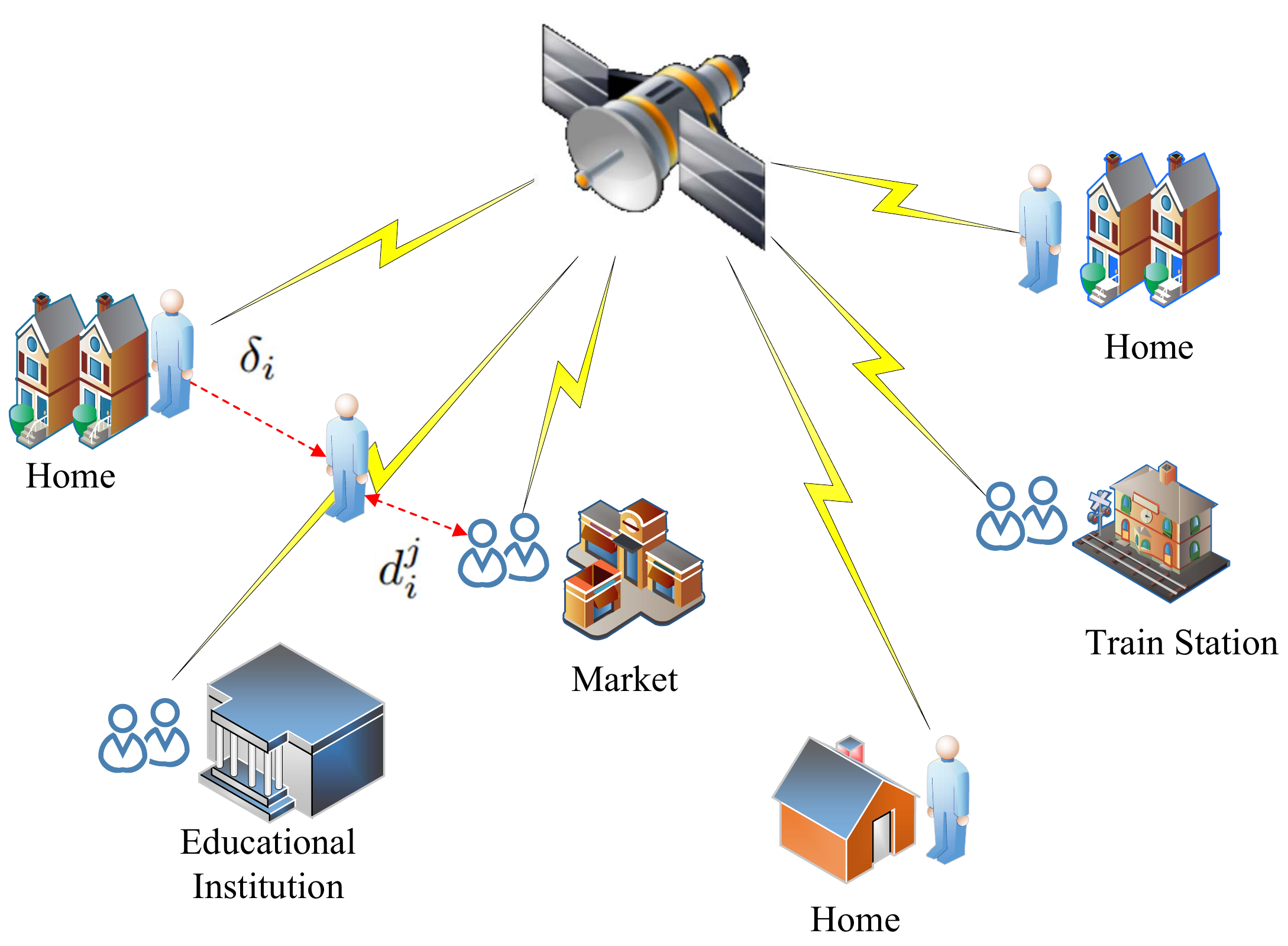}
	\caption{\textcolor{black}{Exemplary System model. Isolation indicates staying at home whereas social distancing measures the distance of a individual from others.} }\label{Fig_System}
	%$\vspace{-.4cm}$
\end{figure}

\section{System Model and Problem Formulation}
\label{Sys_model}
Consider an area in which a set $\mathcal{N}$ of $N$ individuals are living under COVID-19 threat and must decide whether to stay at home or go leave their homes to visit a market, shop, train station, or other locations, as shown in Figure \ref{Fig_System}. Everyone has a mobile phone with GPS. From analyzing the GPS information, we can know their home locations of each individuals, and longitude and latitude of these locations are denoted by $\boldsymbol{X}^h$, and $\boldsymbol{Y}^h$, respectively. We consider one time period (e.g., 15 or 30 minutes) for our scenario and this time period is divided into $T$ smaller time steps in a set $\mathcal{T}$. For each of time step $t\in \mathcal{T}$, we have the GPS coordinates $\boldsymbol{X}$ and $\boldsymbol{Y}$ of every individual.

Now, the deviation from home for any individual $i\in \mathcal{N}$ in between two time steps can be measured by using Euclidean distance as follows:
\begin{equation}
\delta_i^t=
\begin{cases}
\sqrt{(X_i^h-X_i^t)^2+(Y_i^h-Y_i^t)^2}, & \text{if}\  t=1, \\
\sqrt{(X_i^{t-1}-X_i^t)^2+(Y_i^{t-1}-Y_i^t)^2}, & \text{otherwise}.
\end{cases}
\end{equation}   
%The traditional GPS has a maximum error of 10-15 m \cite{Arnold2011}, and hence, we can rewrite $\delta_i^t$  as follows:
%\begin{equation}
%\delta_i^t=
%\begin{cases}
%0, & \text{if}\  \delta_i^t \le 15, \\
%\delta_i^t, & \text{otherwise}.
%\end{cases}
%\end{equation} 
Thus, the total deviation from home by each individual $i\in \mathcal{N}$ in a particular time period can be calculated as follows:
\begin{equation}\label{Equ_delta_i}
\delta_i= \sum_{t\in \mathcal{T}}\delta_i^t, \forall i\in \mathcal{N}
\end{equation} 
%Hence, the grand deviation of all the individuals of $\mathcal{N}$ can be summarized as follows:
%\begin{equation}\label{Equ_delta}
%\delta= \sum_{i\in \mathcal{N}}\delta_i.
%\end{equation}
On the other hand, at the end of a particular time period, the distance between an individual $i\in \mathcal{N}$ and any other individuals $j\in \mathcal{N}, j\ne i$ is as follows:
\begin{equation}
d_i^j=\sqrt{(X_i^{T}-X_j^T)^2+(Y_i^{T}-Y_j^T)^2}.
\end{equation} 
Hence, the total distance of individual $i\in \mathcal{N}$ from other individuals $\mathcal{N}_i\subseteq \mathcal{N}$, who are in close proximity with $i\in \mathcal{N}$, \textcolor{black}{and $\mathcal{N}_i$ is fixed for a particular time step,} can be expressed as follows:
\begin{equation} \label{Equ_d_i}
d_i= \sum_{j\in \mathcal{N}_i}d_i^j,\forall i\in \mathcal{N}.
\end{equation}   
%Thus, we can consider a subset $\mathcal{N}_i$ of $\mathcal{N}$ who are in close proximity with $i\in \mathcal{N}$. So, $d_i$ now looks like as follows:
% \begin{equation}\label{Equ_d_i}
% d_i= \sum_{j\in \mathcal{N}_i}d_i^j,\forall i\in \mathcal{N}
% \end{equation} 
% In the similar fashion of (\ref{Equ_delta}), the grand distance of all $N$ individuals can be summarized as follows:
%\begin{equation}\label{Equ_d}
%d= \sum_{i\in \mathcal{N}}d_i.
%\end{equation} 

Our objective is to keep $\delta$ minimum for reducing the spread of COVID-19 from infected individuals, which is an isolation strategy. Meanwhile, we want to maximize social distancing which mathematically translates into maximizing $d$ for reducing the chance of infection from others. However, we can use $\log$ term to bring fairness \cite{Bairagi2018, Bairagi2019} in the objective function among all individuals. Hence, we can pose the following optimization problem:
\begin{subequations}\label{Opt1}
	\begin{align}
	% \begin{aligned}
	% &\underset{\boldsymbol{y,\alpha}}{\text{max}}
	%\textbf{P1:} \quad
	\textcolor{black}{\underset{\boldsymbol{X},\boldsymbol{Y}} \max} \
	& \textcolor{black}{\underset{i\in \mathcal{N}}\min} \
	\textcolor{black}{\log(Z-\delta_i)^\omega d_i^{(1-\omega)} \tag{\ref{Opt1}}} \\
	\text{s.t.} \quad &\label{Opt1:const1} \delta_i\le \delta_{\textrm{max}},\\ 
	&\label{Opt1:const2} d_i^j\ge d_{\textrm{min}}, \forall i,j\\
	&\label{Opt1:const3} \omega\in [0,1].
	\end{align}
	% \end{equation}
\end{subequations}
In (\ref{Opt1}), $Z$ is a large number for changing the minimization problem to maximization one, and $Z>\delta_i, \forall i\in \mathcal{N}$. The optimization variables $\boldsymbol{X}$ and $\boldsymbol{Y}$ indicate longitude, and latitude, respectively, of the individuals. Moreover, the first term in (\ref{Opt1}) encourages individual for \emph{isolation} whereas the second term in (\ref{Opt1}) encourages individual to maintain fair \emph{social distancing}. In this way, solving (\ref{Opt1}) can play a vital role in our understanding on how to control the spread of COVID-19 among vast population in the society. Constraint (\ref{Opt1:const1}) guarantees small deviation to maintain emergency needs, while Constraint (\ref{Opt1:const2}) assures a minimum fair distance among all the individuals to reduce the spreading of COVID-19 from one individual to another. \textcolor{black}{Usually, local government or authority can set the value of $\delta_{\textrm{max}}$, and  $d_{min}$ can be set up by expertise body like the World Health Organization (WHO).}  Constraint (\ref{Opt1:const3}) shows that $\omega$ can take any value between 0 and 1 which captures the importance between two key factors captured in the objective function of (\ref{Opt1}). For example, if COVID-19 is already spreading in a given society, then most of the weight would go to isolation term rather than social distancing. \textcolor{black}{Here, we present a utility-based model depending on the preventing mechanism like home isolation and social distancing. This is an indirect approach to combat an epidemic like COVID-19.  We have no scope to combine traditional probabilities as described in epidemic models (i.e., SIR, SIRS, SEIR, SIHR, SIDARTHE, etc.) like infection and recovery in our utility-based model.} The objective of (\ref{Opt1}) is difficult to achieve as it requires the involvement and coordination among all the $N$ individual. Moreover, if the individuals are not convinced then it is also difficult for the government to attain the objective forcefully. Thus, we need an alternative solution approach that encourage individual separately to achieve the objective and game theory, which is successfully used in \cite{Bairagi_2018,Abedin2019}, can be one potential solution, which will be elaborated in the next section.

\section{A Noncooperative Game Solution}
\label{Sol_Game}
To attain the objective for a vast population, governments can introduce incentives for isolation and also for social distancing. Then every individual wants to maximize their utilities or payoffs. In this way, government can play its role for achieving social objective. Hence, the modified objective function is given as follows:
\begin{equation}\label{Equ_U}
U(\boldsymbol{\delta},\boldsymbol{d})=\alpha \sum_{i\in \mathcal{N}}\log(Z-\delta_i) +  \beta\sum_{i\in \mathcal{N}}\log d_i,
\end{equation}
where $\alpha=\alpha^{'}\omega$ and $\beta=\beta^{'}(1-\omega)$ with $\alpha^{'}>0$ and $\beta^{'}>0$ are incentives per unit of isolation and social distancing. In practice, $\alpha$ and $\beta$ can be monetary values for per unit of isolation and social distancing, respectively. In (\ref{Equ_U}), one individual's position affects the social distancing of others, and hence, the individuals have partially conflicting interest on the outcome of $U$. Therefore, the situation can be interpreted with the noncooperative game \cite{Wang2014, Song2012}.   

A noncooperative game is a game that exhibit a competitive situation where each \textcolor{black}{player (i.e., individual)} needs to make choices independent of the other individuals, given the possible policies of the other individuals and their impact on the individual's payoffs or utilities. Now, a noncooperative game in strategic form or a strategic game $\mathcal{G}$ is a triplet $\mathcal{G}=(\mathcal{N}, (\mathcal{S}_i)_{i\in \mathcal{N}}, (u_i)_{i\in \mathcal{N}})$ \cite{Han2012} for any time period  where:
\begin{itemize}
	\item $\mathcal{N}$ is a finite set of individuals, i.e., $\mathcal{N}=\{1,2,\cdots,N\}$,
	\item $\mathcal{S}_i$ is the set of available strategies for individual $i\in \mathcal{N}$,
	\item $u_i:\mathcal{S}\rightarrow \mathbb{R}$ is the payoff function of individual $ i\in \mathcal{N}$, with $\mathcal{S}=\mathcal{S}_1\times \mathcal{S}_2\times..\times \mathcal{S}_N$. 
\end{itemize}
\textcolor{black}{Theoretically, there may be many different strategies, but in our case, we have considered two strategies, namely staying at home and moving outside (visit market, shop, train station, school/college etc.) for every individual.} We use  $\mathcal{S}_i=\{s_i^h,s_i^m\}$ to represent the set of strategies for each individual $i$ where $s_i^h$ and $s_i^m$ indicate the strategies of staying at home and moving outside for individual $i\in \mathcal{N}$, respectively. The payoff or incentive function of any individual $i\in \mathcal{N}$ in a time period can be defined as follows:
\begin{equation}
u_i(.) =
\begin{cases}
\alpha \log Z +\beta \log \tilde{d}_i, & \text{if strategy is}\  s_i^h, \\
\alpha \log (Z-\delta_i)+\beta \log d_i, & \text{if strategy is}\  s_i^m.
\end{cases}
\end{equation}
where $\tilde{d}_i=\sum_{j\in \mathcal{N}_i}\sqrt{(X_i^h-X_j)^2+(Y_i^h-Y_j)^2}$.

The Nash equilibrium \cite{Nash1950} is the most used solution concept for a noncooperative game. Formally, Nash equilibrium can be defined as follows \cite{Basar1999}:
\begin{definition}:
	\label{Def_NE}
	A pure strategy Nash equilibrium for a non-cooperative game $\mathcal{G}=(\mathcal{N}, (\mathcal{S}_i)_{i\in \mathcal{N}}, (u_i)_{i\in \mathcal{N}})$ is a strategy profile $\mathbf{s}^*\in \mathcal{S}$ where
	$u_i(s_i^*,{\boldsymbol{s}}_{-i}^*)\ge u_i(s_i,{\boldsymbol{s}}_{-i}^*), \forall s_i \in \mathcal{S}_i,\forall i\in \mathcal{N}$.
\end{definition}

However, to find the Nash equilibrium, the following two definitions can be helpful.
\begin{definition}\cite{Han2012}:
	\label{Def_dominant_strategy}
	A strategy $s_i\in \mathcal{S}_i$ is the dominant strategy for individual $i\in \mathcal{N}$ if $u_i(s_i,{s}_{-i})\ge u_i(s_i^{'},{s}_{-i}), \forall s_i^{'} \in \mathcal{S}$ and $\forall s_{-i} \in \mathcal{S}_{-i}$, where $\mathcal{S}_{-i}=\prod_{j\in \mathcal{N}, j\ne i}\mathcal{S}_j$ is the set of all strategy profiles for all individuals except $i$.
	
\end{definition}

\begin{definition}\cite{Han2012}:
	\label{Def_dominant_strategy_equilibrium}
	A strategy profile $\mathbf{s}^*\in \mathcal{S}$ is the dominant strategy equilibrium if every elements $s_i^*$ of $\mathbf{s}^*$ is the dominant strategy of individual $i\in \mathcal{N}$.
\end{definition}

Thus, if we can show that every individual of our game $\mathcal{G}$ has a strategy that gives better utility irrespective of other individuals strategies, then with the help of Definition \ref{Def_dominant_strategy} and \ref{Def_dominant_strategy_equilibrium}, we can say that Proposition \ref{Theo_dominant_strategy} is true.

\begin{proposition}
	\label{Theo_dominant_strategy}
	$\mathcal{G}$ has a pure strategy Nash equilibrium when $\alpha>\beta$.  
\end{proposition}
\begin{proof}[Proof]
	Let us consider a 2-individuals simple matrix game as shown in Table \ref{Table_game} with the mentioned strategies. \textcolor{black}{For simplicity, we consider Laplacian distance $\Delta$ that each individual can pass in any timestamp.}	
	\begin{table}
		\setlength{\extrarowheight}{2pt}
		\centering
		\caption{Game matrix for 2-individuals}\label{Table_game}
		\begin{tabular}{cc|c|c|}
			& \multicolumn{1}{c}{} & \multicolumn{2}{c}{P $2$}\\
			& \multicolumn{1}{c}{} & \multicolumn{1}{c}{$s_2^h$}  & \multicolumn{1}{c}{$s_2^m$} \\\cline{3-4}
			\multirow{2}*{P $1$}  & $s_1^h$ & $(u_1(s_1^h,s_2^h),u_2(s_1^h,s_2^h))$ & $(u_1(s_1^h,s_2^m),u_2(s_1^h,s_2^m))$ \\\cline{3-4}
			& $s_1^m$ & $(u_1(s_1^m,s_2^h),u_2(s_1^m,s_2^h))$ & $(u_1(s_1^m,s_2^m),u_2(s_1^m,s_2^m))$ \\\cline{3-4}
		\end{tabular}
	\end{table}
	
	Thus, the utilities of $P_1$:
	\begin{equation}
	%\begin{split}
	\begin{aligned}
	u_1(s_1^h,s_2^h)&=\alpha \log Z + \beta\log d_1,\\
	u_1(s_1^h,s_2^m)&=\alpha \log Z + \beta\log (d_1\pm\Delta),\\
	u_1(s_1^m,s_2^h)&=\alpha \log (Z-\Delta) + \beta\log (d_1\pm\Delta),\\
	u_1(s_1^m,s_2^m)&=\alpha \log (Z-\Delta) + \beta\log (d_1\pm 2\Delta),
	\end{aligned}\label{Eq_payoff_1}
	%\end{split}
	\end{equation}
	where $\pm$ indicates the movement of individual to other individual and opposite direction, respectively. Now,
	\begin{equation}
	%\begin{split}
	\begin{aligned}
	u_1(s_1^h,s_2^h)-u_1(s_1^m,s_2^h) =&\\ \alpha \log \bigg(\frac{Z}{Z-\Delta}\bigg) + \beta\log \bigg(\frac{d_1}{d_1\pm\Delta}\bigg),\\
	u_1(s_1^h,s_2^m)-u_1(s_1^m,s_2^m) =& \\ \alpha \log \bigg(\frac{Z}{Z-\Delta}\bigg) + \beta\log \bigg(\frac{d_1\pm\Delta}{d_1\pm 2\Delta}\bigg).	
	\end{aligned}\label{Eq_payoff_1_1}
	%\end{split}
	\end{equation}
	As $\alpha>\beta$, so the following conditions hold from (\ref{Eq_payoff_1_1}):
	\begin{equation}
	%\begin{split}
	\begin{aligned}
	u_1(s_1^h,s_2^h)-u_1(s_1^m,s_2^h) \ge 0,\\
	u_1(s_1^h,s_2^m)-u_1(s_1^m,s_2^m) \ge 0,	
	\end{aligned}\label{Eq_payoff_1_2}
	%\end{split}
	\end{equation}
	Hence, rewriting (\ref{Eq_payoff_1_2}), we get the followings:
	\begin{equation}
	%\begin{split}
	\begin{aligned}
	u_1(s_1^h,s_2^h)\ge u_1(s_1^m,s_2^h),\\
	u_1(s_1^h,s_2^m)\ge u_1(s_1^m,s_2^m).
	\end{aligned}\label{Eq_payoff_1_3}
	%\end{split}
	\end{equation}
	Thus, $s_1^h$ is the dominant strategy of $P_1$. Moreover, for the individual $P_2$, the utilities are as follows:
	\begin{equation}
	%\begin{split}
	\begin{aligned}
	u_2(s_1^h,s_2^h)&=\alpha \log Z + \beta\log d_2,\\
	u_2(s_1^m,s_2^h)&=\alpha \log Z + \beta\log (d_2\pm\Delta),\\
	u_2(s_1^h,s_2^m)&=\alpha \log (Z-\Delta) + \beta\log (d_2\pm\Delta),\\
	u_2(s_1^m,s_2^m)&=\alpha \log (Z-\Delta) + \beta\log (d_2\pm 2\Delta).
	\end{aligned}\label{Eq_payoff_2}
	%\end{split}
	\end{equation}
	Now,
	\begin{equation}
	%\begin{split}
	\begin{aligned}
	u_2(s_1^h,s_2^h)-u_2(s_1^h,s_2^m) =&\\ \alpha \log \bigg(\frac{Z}{Z-\Delta}\bigg) + \beta\log \bigg(\frac{d_2}{d_2\pm\Delta}\bigg),\\
	u_2(s_1^m,s_2^h)-u_2(s_1^m,s_2^m) =&\\ \alpha \log \bigg(\frac{Z}{Z-\Delta}\bigg) + \beta\log \bigg(\frac{d_2\pm\Delta}{d_2\pm 2\Delta}\bigg).
	\end{aligned}\label{Eq_payoff_2_1}
	%\end{split}
	\end{equation}
	As $\alpha>\beta$, so the following conditions hold from (\ref{Eq_payoff_2_1}):
	\begin{equation}
	%\begin{split}
	\begin{aligned}
	u_2(s_1^h,s_2^h)-u_2(s_1^m,s_2^m) \ge 0,\\
	u_2(s_1^m,s_2^h)-u_2(s_1^m,s_2^m) \ge 0,	
	\end{aligned}\label{Eq_payoff_2_2}
	%\end{split}
	\end{equation}
	Hence, rewriting (\ref{Eq_payoff_2_2}), we get the followings:
	\begin{equation}
	%\begin{split}
	\begin{aligned}
	u_2(s_1^h,s_2^h)\ge u_2(s_1^h,s_2^m),\\
	u_2(s_1^m,s_2^h)\ge u_2(s_1^m,s_2^m).	
	\end{aligned}\label{Eq_payoff_2_3}
	%\end{split}
	\end{equation}
	Therefore, $s_2^h$ is the dominant strategy of individual $P_2$.	
	
	\textcolor{black}{When there are $N$-individuals ($N>2$) in the game, every individual still has the same two strategies as in the case of a 2-individuals game. However, the dimension of the game matrix, which is actually representing the payoff of every individual for each strategy, will be changed from $2\times2$ to $2^N$.} The incentive of individual $i\in \mathcal{N}$ (takes strategy $s_i^h$, without considering others strategy), is given as follows:
	\begin{equation}
	u_i(s_i^h,\dots) = \alpha \log Z +\beta \log \tilde{d}_i.
	\end{equation}
	However, if the individual $i\in \mathcal{N}$ takes the strategy $s_i^m$, i.e., the individual visits some crowded place like market, shop, train station, school, or other location, then a person may come in close contact with many others. Thus, the incentive of individual $i\in \mathcal{N}$ with this strategy is given as follows:
	\begin{equation}
	u_i(s_i^m,\dots) = \alpha \log (Z-\delta_i) +\beta \log d_i,
	\end{equation}
	where $\delta_i$ is calculated from (\ref{Equ_delta_i}) and $d_i$ is measured from (\ref{Equ_d_i}) for that particular location. Moreover, $d_i<\tilde{d}_i$ as these places are crowded and individuals are in short distance with one another. Hence, $u_i(s_i^h,\dots)>u_i(s_i^m,\dots)$ as $Z>Z-\delta_i$ and $\tilde{d}_i>d_i$ for any individual $i\in \mathcal{N}$. That means, $s_i^h$ is the dominant strategy for individual $i\in \mathcal{N}$ irrespective of the strategies of other individuals in the game $\mathcal{G}$.
	Thus, there is a strategy profile $\boldsymbol{s}^*=\{s_1^h,s_2^h,\cdots,s_N^h\}\in \mathcal{S}$ where each element $s_i^*$ is a dominant strategy. Hence, by Definition \ref{Def_dominant_strategy_equilibrium}, $\boldsymbol{s}^*$ is a dominant strategy equilibrium. Moreover, a dominant strategy equilibrium is always a Nash equilibrium \cite{Han2012}. Hence, the game $\mathcal{G}$ has always a pure strategy Nash equilibrium.
\end{proof} 

%\cite{domi_nash}

Thus, Nash equilibrium is the solution of the noncooperative game $\mathcal{G}$. In this equilibrium, no individual of $\mathcal{N}$ has the benefit of changing their strategy while others remain in their strategies. That means, the utility of each individual $i\in \mathcal{S}$ is maximized in this strategy, and hence ultimately maximize the utility of (\ref{Equ_U}). In fact, incentivizing the social distancing mechanism is promoting social distancing to each individual. To this end, maximizing $U$ of (\ref{Equ_U}) ultimately maximize the original objective function of (\ref{Opt1}).

Moreover, the Nash equilibrium point has a greater implication on controlling the spread of COVID-19 in the society. At the NE point, every individual stays at home, \textcolor{black}{and that is the only NE point in our game environment}. So, if someone gets affected by COVID-19, the individual will not go in contact with others. Similarly, an unaffected individual has no probability to come in contact with an affected individuals. Unfortunately, the family members have the chance to be affected if they don't follow fair distance and health norms.

\textcolor{black}{For calculating the utility of each player (i.e., individual) $i\in \mathcal{N}$ in case of any strategy, the positional information of $\mathcal{N}_i$ individuals who are in close proximity of $i\in \mathcal{N}$ are necessary. In order to obtain this information, each individual $i\in \mathcal{N}$ communicates with GPS satellite and it sends the information of $\mathcal{N}_i, \forall i\in \mathcal{N}$  individuals to the corresponding player $i\in \mathcal{N}$.  Hence, GPS satellite needs to send $\sum_{i\in \mathcal{N}}|\mathcal{N}_i|=C|\mathcal{N}|$ information, where $C$ can be a fixed number to represent the close-proximity individuals and equal for every  $i\in \mathcal{N}$, as a whole for $|\mathcal{N}|$ individuals. Thus,  the complexity of the game is proportional to the number of players $|\mathcal{N}|$ of the game, and will not increase exponentially.}

\section{Sustainability of Lockdown Policy with the System Model}
\label{Sustainability}
The sustainability of the lockdown policy can be interpreted by using the outcome of the Nash equilibrium point that is achieved in the noncooperative game in Section \ref{Sol_Game}.

The total amount of incentive a particular time period is presented in (\ref{Equ_U}). In a particular day, we have $T_s=\frac{24\times 60}{T_0}$ time period where $T_0$ is the length of a time period in minutes. Thus, we can denote the incentive of a time stamp $t_s$ in a particular day $p$ as follows:
\begin{equation}\label{Equ_U10}
U_p^{t_s}(\boldsymbol{\delta},\boldsymbol{d})=\alpha \sum_{i\in \mathcal{N}}\log(Z-\delta_i) +  \beta\sum_{i\in \mathcal{N}}\log d_i.
\end{equation}
Hence, the amount of resources/money that is necessary to incentivize individuals in a particular day, $p$ can be expressed as follows:
\begin{equation}\label{Equ_U1}
U_p= \sum_{t_s=1}^{T_s}U_p^{t_s}(\boldsymbol{\delta},\boldsymbol{d}).
\end{equation}

Now, if we are interested to find the sustainability of lockdown policy for a particular country till a certain number of days, denoted by $P$, we have to satisfy the following inequality:
\begin{equation}\label{Equ_U1_2}
\sum_{p=1}^{P}U_p \le R_0 + \sum_{p=1}^{P}r_p,
\end{equation}
where $R_0$ is the amount of resource/money of a particular country at the starting of lockdown policy that can be used as incentive and $r_p$ is the collected resources in a particular day, $p$, of the lockdown period. Here, $r_p$ includes governmental revenue and donation from different individuals, organizations and even countries. Moreover, the unit of $\alpha$, $\beta$, $R_0$ and $r_p$ are same.

If we assume for simplicity that $U_p$ and $r_p$ are same for every day and they are denoted by $\tilde{U}$ and $\tilde{r}$, respectively, then we can rewrite (\ref{Equ_U1_2}) as follows:  
\begin{equation}\label{Equ_U1_3}
P\times \tilde{U} \le R_0 + P\times \tilde{r}.
\end{equation}
Hence, if we are interested to find the upper limit of sustainable days for a particular country using lockdown policy, then we have the following equality:
\begin{equation}\label{Equ_U1_4}
P\times \tilde{U} = R_0 + P\times \tilde{r}.
\end{equation}
Thus, by simplifying (\ref{Equ_U1_4}), we have the following:
\begin{equation}\label{Equ_U1_5}
P = \frac{R_0}{\tilde{U}-\tilde{r}}.
\end{equation}

Here, the sustainable days $P$ depends on $R_0$, $\tilde{U}$, and $\tilde{r}$. However, we cannot change $R_0$ but government can predict $\tilde{r}$. Moreover, depending on $R_0$ and $\tilde{r}$, government can formulate its policy to set $\alpha$ and $\beta$ so that individuals are encouraged to follow the lockdown policy. Alongside, we cannot continue lockdown policy infinitely based upon the limited total resources. Hence, the governments should formulate and update its lockdown policy based on the predicted sustainable capability to handle COVID-19, otherwise resource crisis will be a further bigger worldwide pandemic.

\section{Numerical Analysis}
\label{Per_eva}
In this section, we assess the proposed approach using numerical analyses. We consider an area of $1,000$ m $\times$ $1,000$ m for our analysis where individuals' position are randomly distributed. The value of the principal simulation parameters are shown in the Table \ref{Table1}.
\begin{table}
	\centering
	\renewcommand{\arraystretch}{1.5}
	\caption{Value of the principal simulation parameters}\label{Table1}
	\begin{tabular}{|c|c|} \hline
		$\textbf{Symbol}$ &$\textbf{Value}$\\ \hline
		$N$ & $\{500, 1000, 1500, 2000\}$\\ \hline
		$\alpha$ & $3.0$  \\ \hline
		$\beta$ & $1.0$ \\ \hline
		$\omega$ & $[0,1]$\\ 
		\hline
		$Z$ & $1400$ \\ \hline
		$R_0$ & \textcolor{black}{$\{5\times 10^{23},5.5\times 10^{23},6\times 10^{23},6.5\times 10^{23},7\times 10^{23}\}$}\\ 
		\hline
		$\tilde{r}$ & \textcolor{black}{$\{\tilde{U}\times0.10, \tilde{U}\times0.20, \tilde{U}\times0.30, \tilde{U}\times0.40, \tilde{U}\times0.50\}$}\\ 
		\hline
	\end{tabular}
	\renewcommand{\arraystretch}{1}
	%$\vspace{-.4cm}$
\end{table}

Figure \ref{Fig_payoff_omega} illustrates a comparison between home isolation (stay at home) and random location in the considered area for a varying value of $\omega$. In this figure, we consider two cases of $N = 500$ and $N=1,000$. In both the cases, home isolation (quarantine) is beneficial over staying in random location and the differences between two approaches are increased with the increasing value of $\omega$. Moreover, the difference of payoffs between two approaches are increased with the increasing value of $\omega$ as the more importance are given in home isolation.     
\begin{figure}
	\centering
	\includegraphics[width=0.5\textwidth]{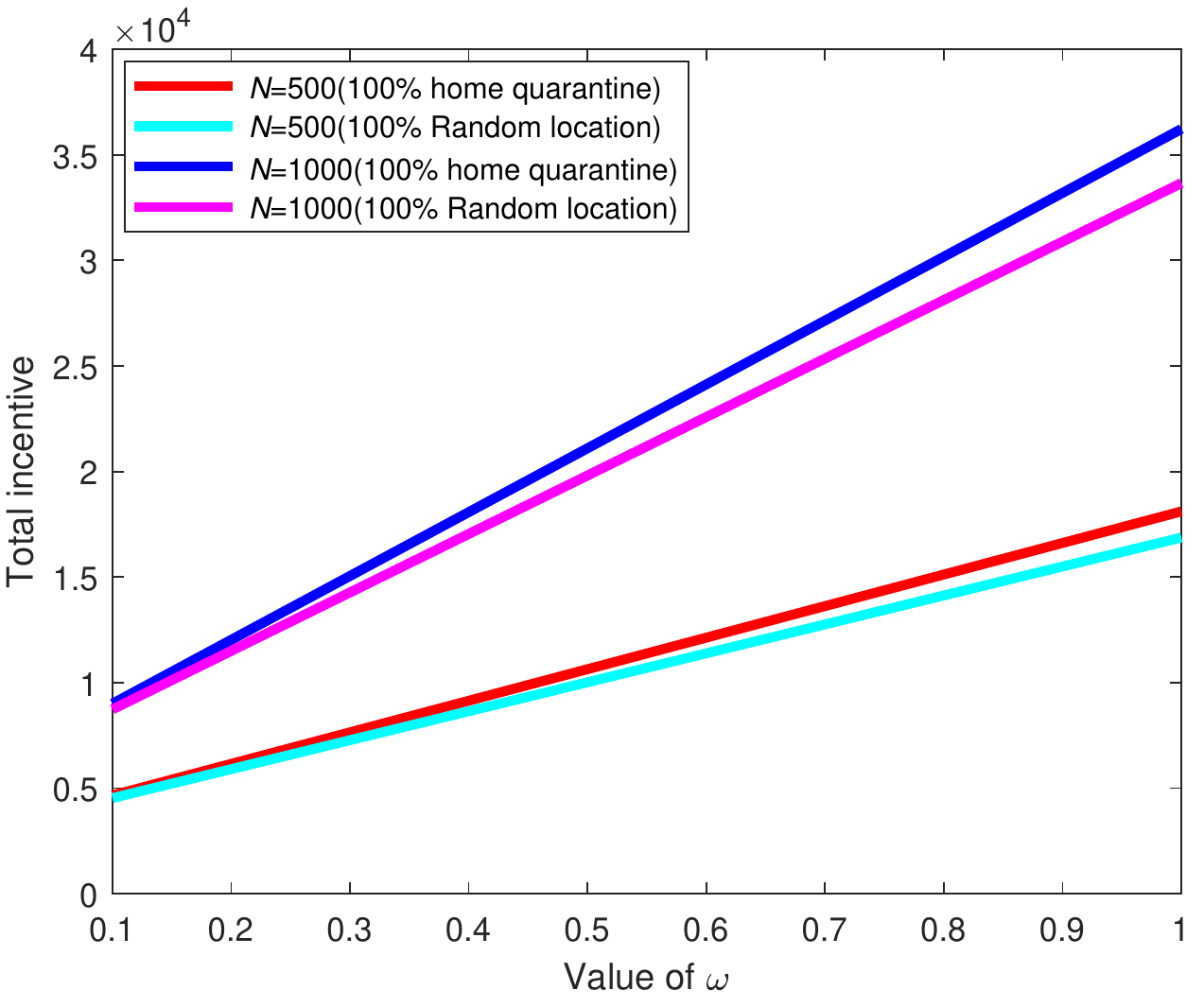}
	\caption{Comparison of incentive (in log scale) for varying value of $\omega$.}\label{Fig_payoff_omega}
	%$\vspace{-.4cm}$
\end{figure}

\begin{figure}
	\begin{subfigure}{.5\textwidth}
		\centering
		\includegraphics[width=.8\linewidth,height=5cm]{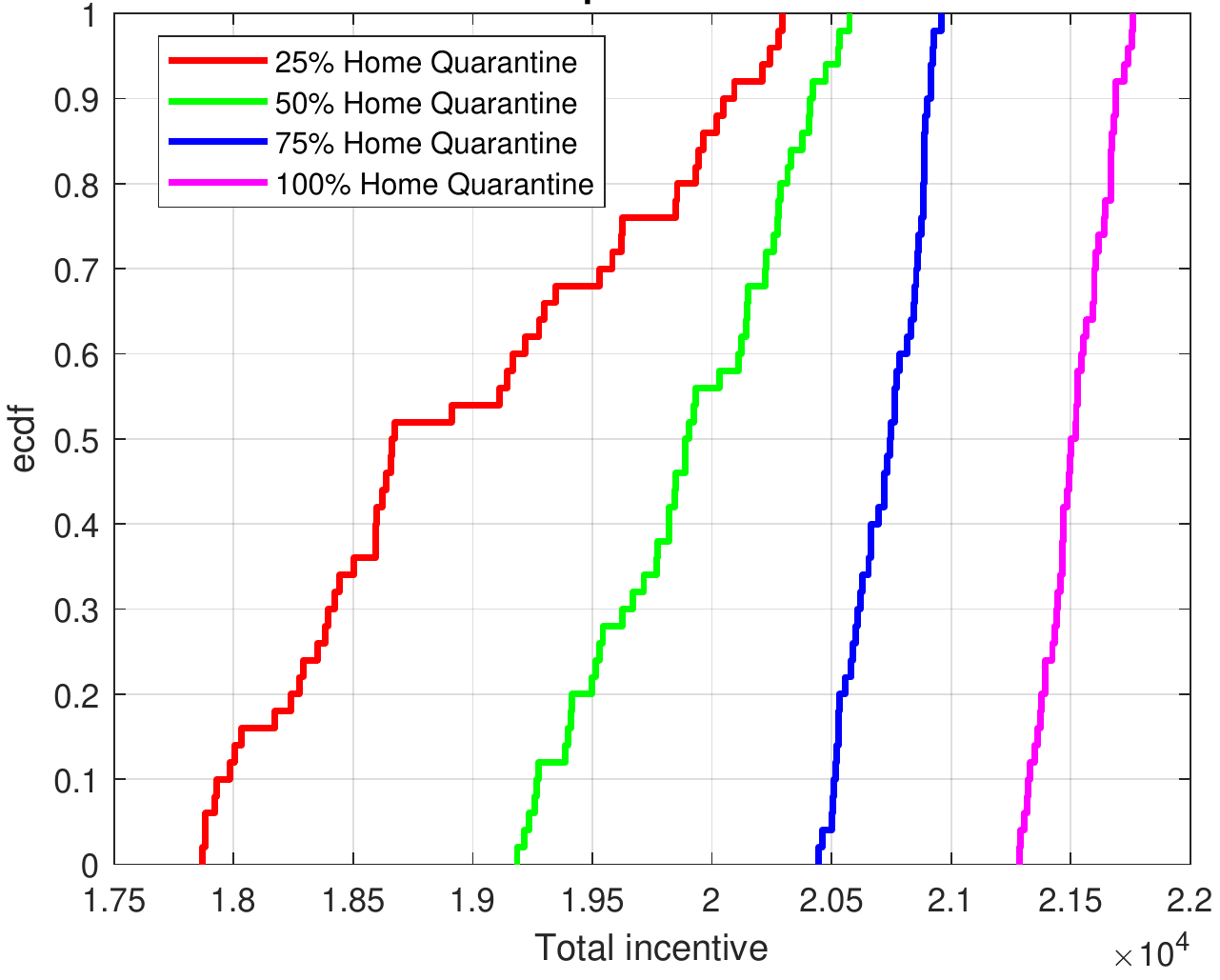}
		\caption{$N=500$}
		\label{fig:ecdf500}
	\end{subfigure}%
	\\
	\begin{subfigure}{.5\textwidth}
		\centering
		\includegraphics[width=.8\linewidth,height=5cm]{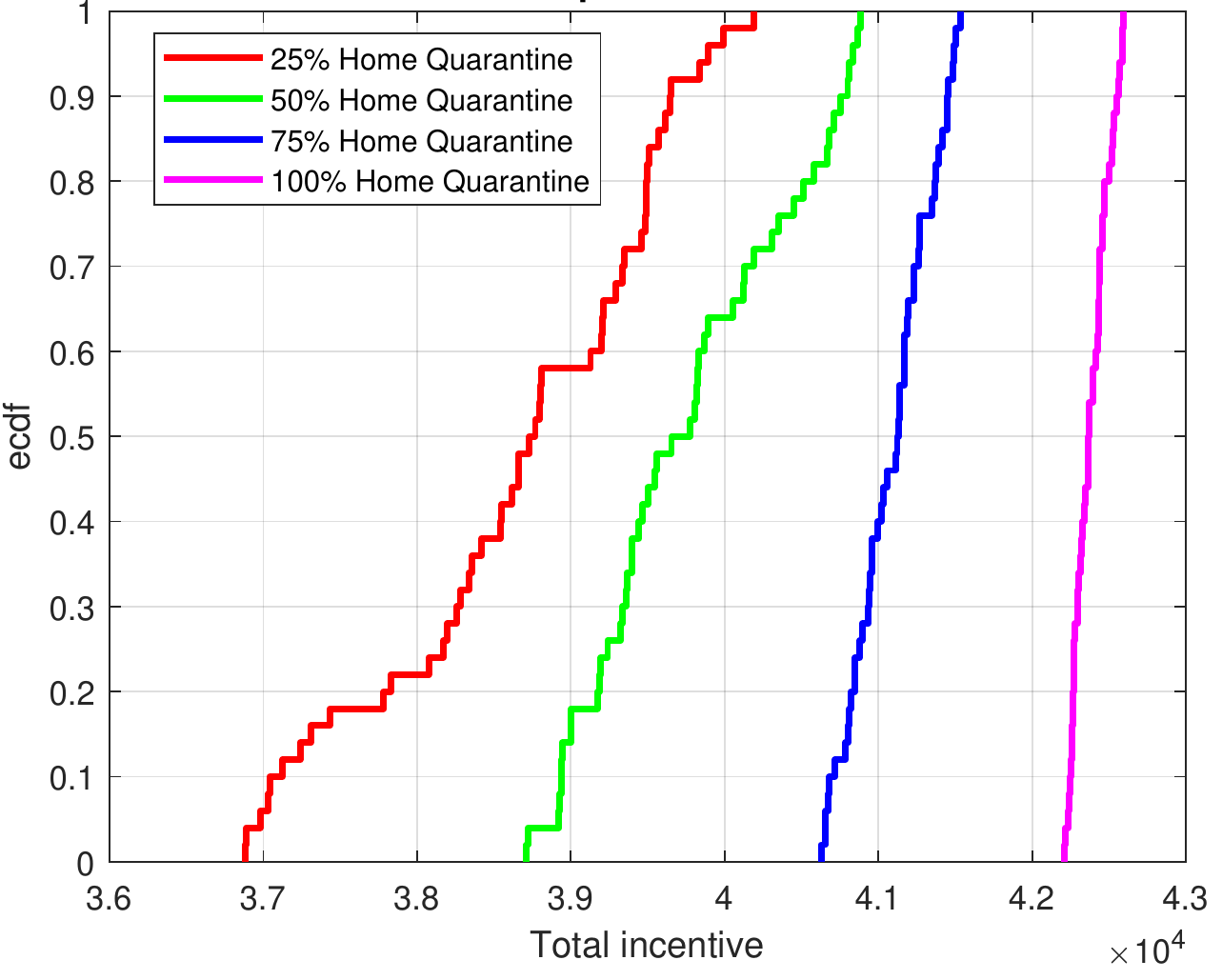}
		\caption{$N=1,000$}
		\label{fig:ecdf1000}
	\end{subfigure}
	\\
	\begin{subfigure}{.5\textwidth}
		\centering
		\includegraphics[width=.8\linewidth,height=5cm]{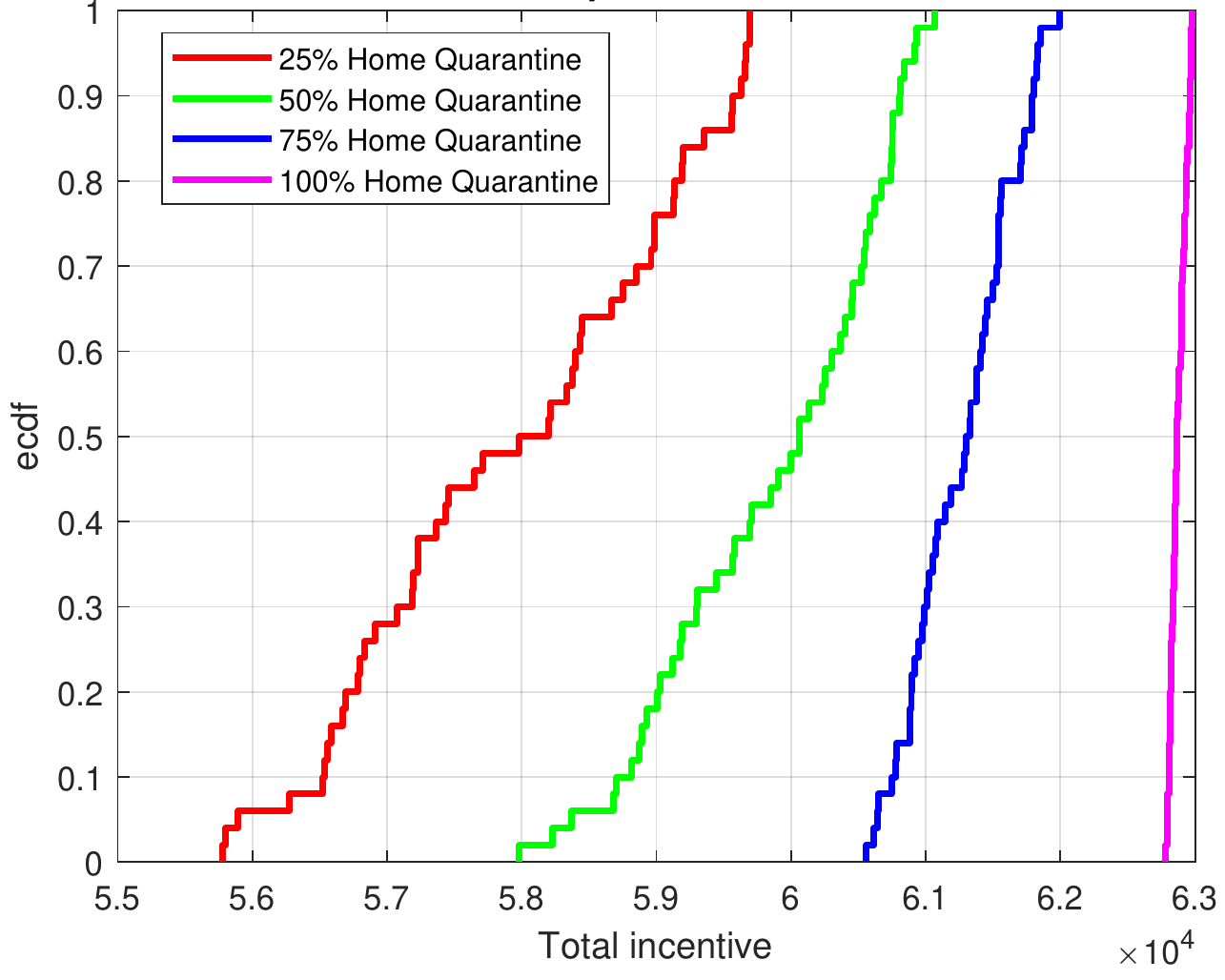}
		\caption{$N=1,500$}
		\label{fig:ecdf1500}
	\end{subfigure}%
	\\
	\begin{subfigure}{.5\textwidth}
		\centering
		\includegraphics[width=.8\linewidth,height=5cm]{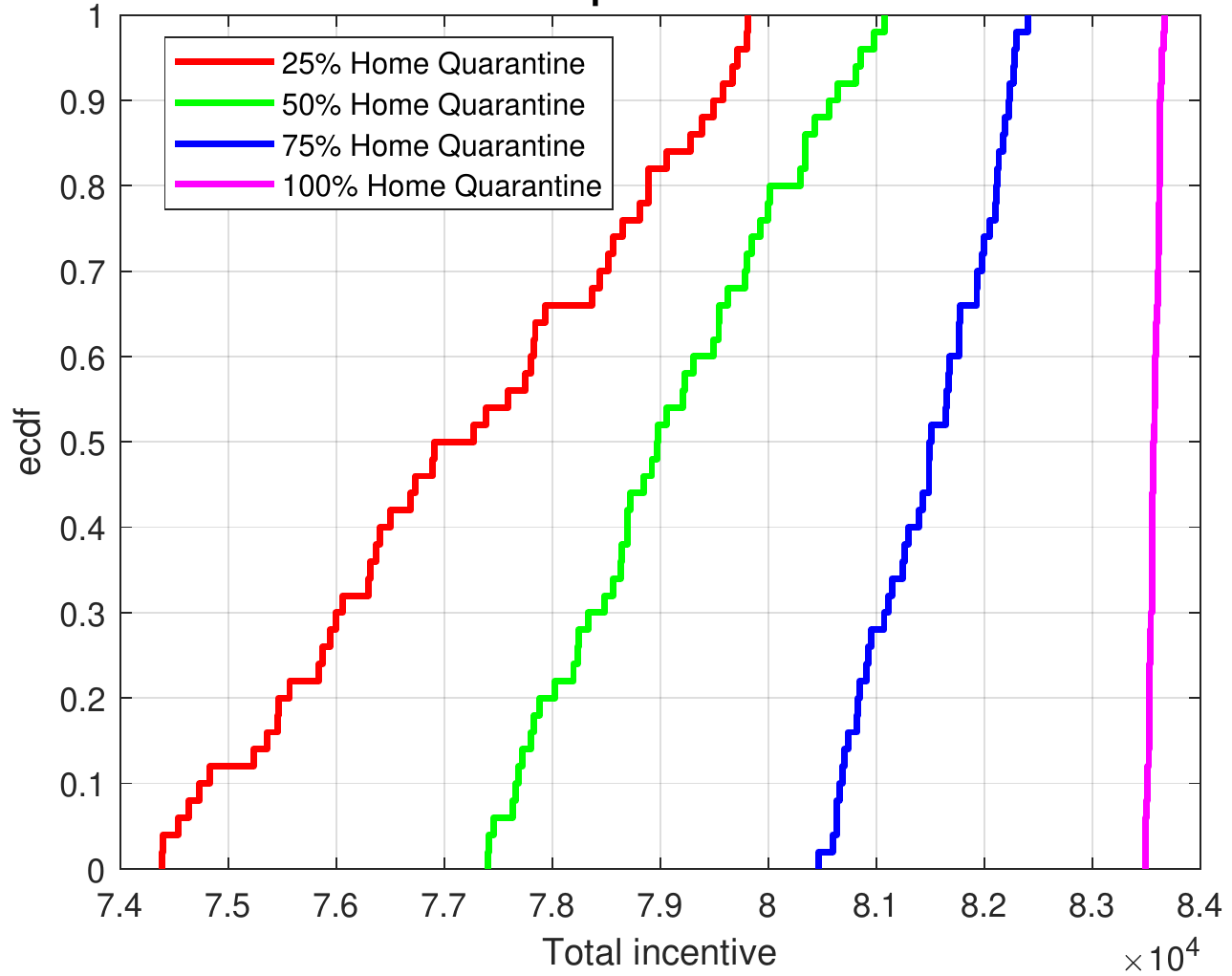}
		\caption{$N=2,000$}
		\label{fig:ecdf2000}
	\end{subfigure}%
	\caption{Ecdf of incentives (in log scale) for different value of $N$ with $\alpha=3.0$ and $\beta=1.0$ using $50$ runs.}
	\label{fig:ecdf}
\end{figure}

Figure \ref{fig:ecdf} shows the empirical cumulative distribution function (ecdf) of incentives for different numbers of individuals. The figure revels that the incentive values increase with the increasing number of home quarantine individuals in all the four cases. Figure \ref{fig:ecdf500} exhibits that the incentives are below $19,000$, and $20,000$ for $50\%$, and $48\%$ sure, respectively, for $25\%$ and $50\%$ home quarantine cases whereas the incentives are $90\%$ sure in between $20,500$ and $21,000$ for $75\%$ home isolation case. Further, the same values are at least $21,500$ for $50\%$ sure in case of full home isolation. Figure \ref{fig:ecdf1000} depicts that the incentive of being below $38,000$ is $40\%$ sure for $25\%$ home isolation case, however, the same values of being above $40,000$, and $41,000$ are $40\%$, and $60\%$,  sure, respective, for $50\%$, and $75\%$ cases. Moreover, for $100\%$ home isolation case, the values are in between $42,000$ to $43,000$ for sure. The incentives for $25\%$, $50\%$, $75\%$, and $100\%$ home isolation cases are above $57,000$, $59,000$, $61,000$, and $63,000$, respectively, with probability $0.60$, $0.65$, $0.65$, and $0.80$, respectively, as shown in Figure \ref{fig:ecdf1500}. Additionally, the same values are at least $77,000$, $79,000$, $81,000$, and $83,500$ with $0.50$, $0.50$, $0.72$, and $1.00$ probabilities, respectively, which is presented in Figure \ref{fig:ecdf2000}.
\begin{figure}
	\centering
	\includegraphics[width=.5\textwidth]{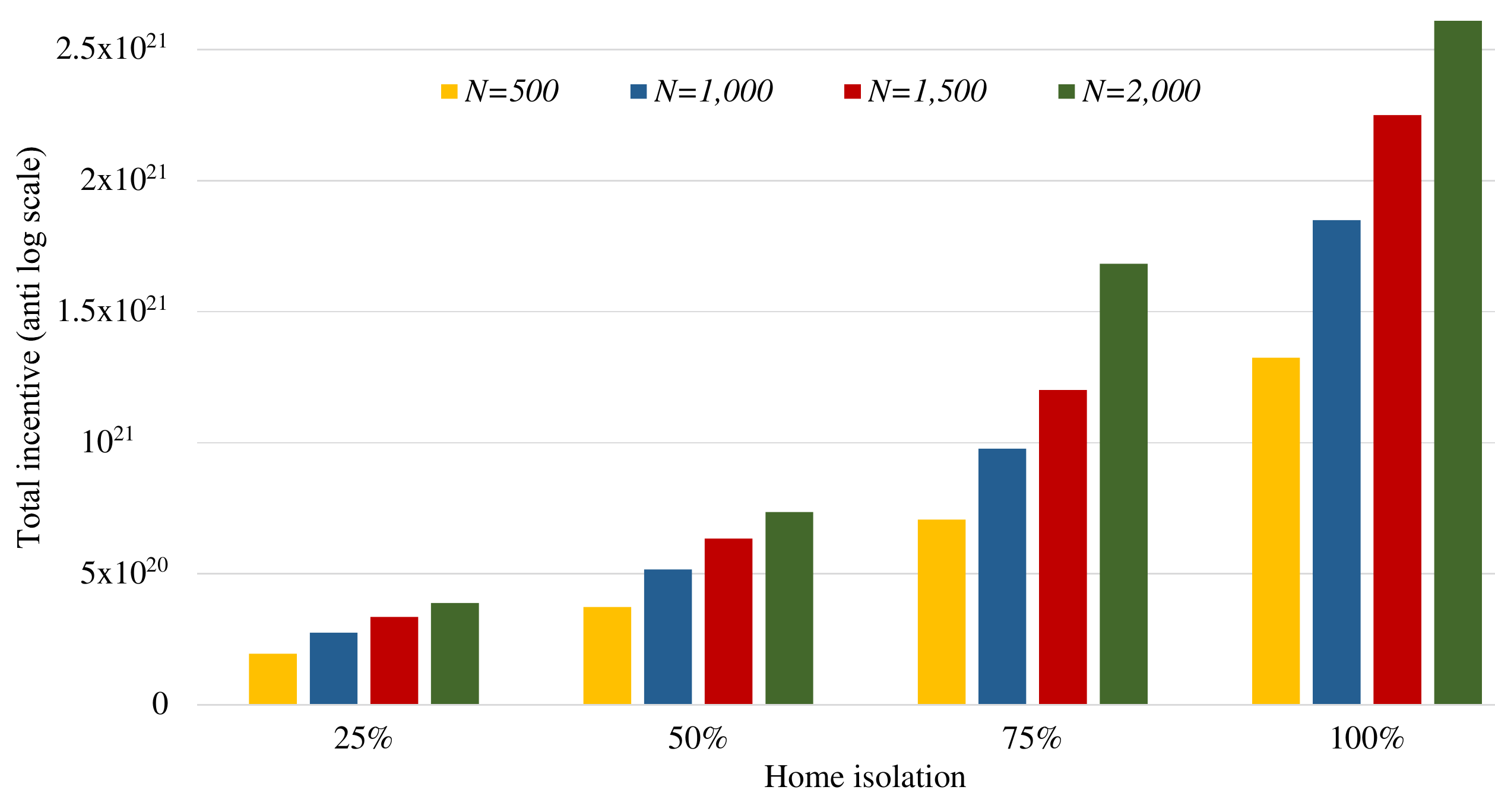}
	\caption{\textcolor{black}{Total incentive (average of 50 runs) for varying percentage of home isolation individuals when $\alpha=3.0$ and $\beta=1.0$.}}\label{Fig_bar}
	%$\vspace{-.4cm}$
\end{figure}

\begin{figure}
	\centering
	\includegraphics[width=.5\textwidth]{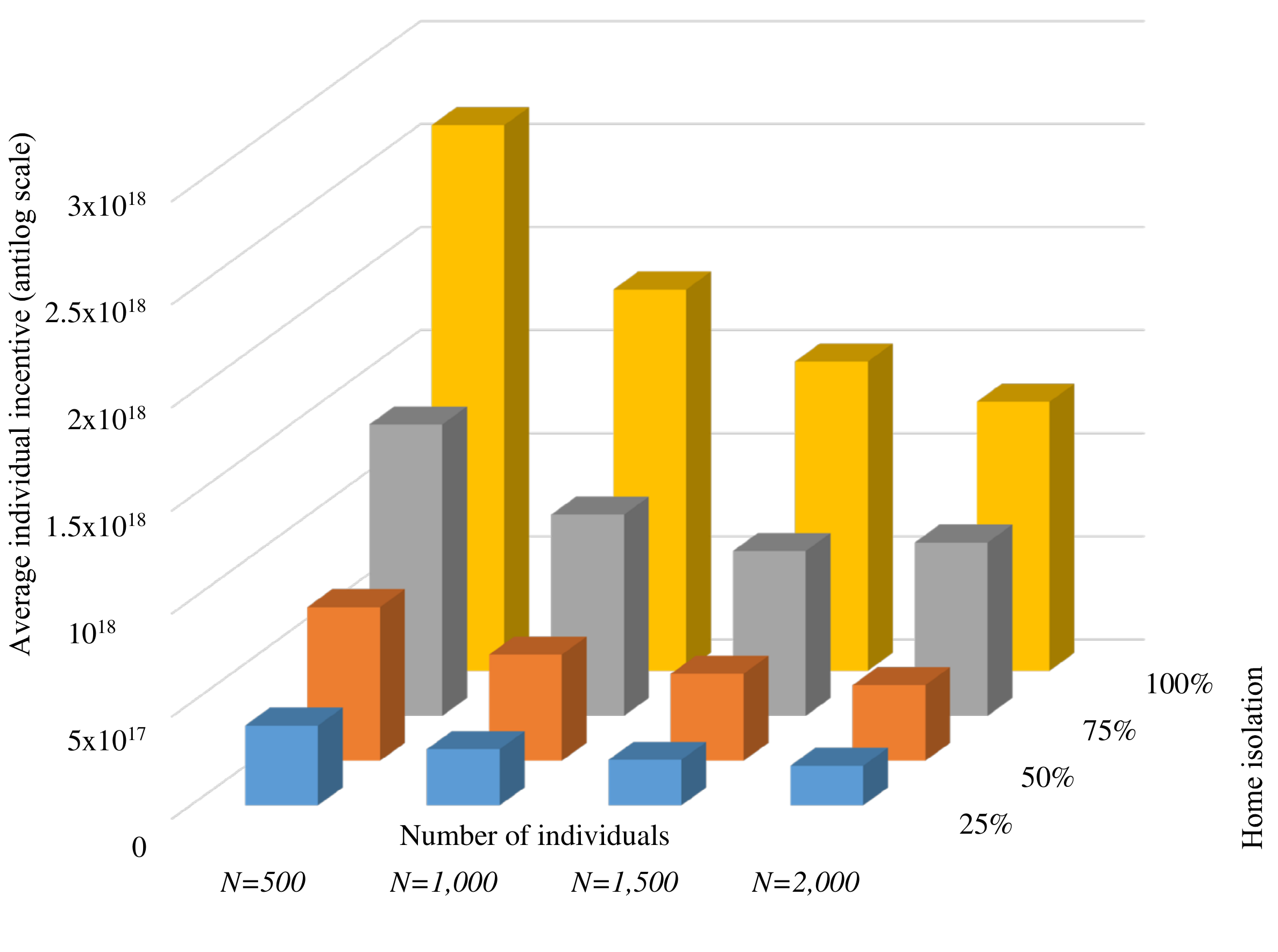}
	\caption{\textcolor{black}{Average individual incentive for varying percentage of home quarantine individuals when $\alpha=3.0$ and $\beta=1.0$.}}\label{Fig_ind_payoff}
	%$\vspace{-.4cm}$
\end{figure}
The total incentive (averaging of 50 runs) for varying percentage of home isolation individuals with different sample size are shown in Figure \ref{Fig_bar}. From this figure, we observe that the total payoff increases with increasing number of home isolation individuals for all considered cases. The incentives are $578\%$, $571\%$, $571\%$, and $571\%$ better from home quarantine of $25\%$ to $100\%$ for $N=500$, $N=1,000$, $N=1,500$, and $N=2,000$, respectively.  Moreover, for a particular percentage of home isolation, the total incentive is related with the sample size. In case of $50\%$ individuals in the home isolation, the incentive for $N=2,000$ is $97.08\%$, $42.50\%$, and $15.96\%$ more than that of $N=500$, $N=1,000$, and $N=1,500$, respectively. 

Figure \ref{Fig_ind_payoff} shows the average individual payoff for varying parentage of home isolation individuals for different scenarios. The figure exhibits that the average individual incentive increases with an increasing percentage of home isolation as the deviation $\boldsymbol{\delta}$ decreases and hence, the value of home isolation incentive increases. For $N=500$, the incentive of $100\%$ home isolation is $85.25\%$ more than that of $25\%$ home isolation. Moreover, in a particular percentage of home isolation, the incentive decreases with an increasing number of considered individuals as the social distancing decreases due to the more number of individuals. In case of $50\%$ home isolation, the individual incentive for $N=500$ is $102.96\%$ more than that of $N=2,000$.

\begin{figure}
	\centering
	\includegraphics[width=.5\textwidth]{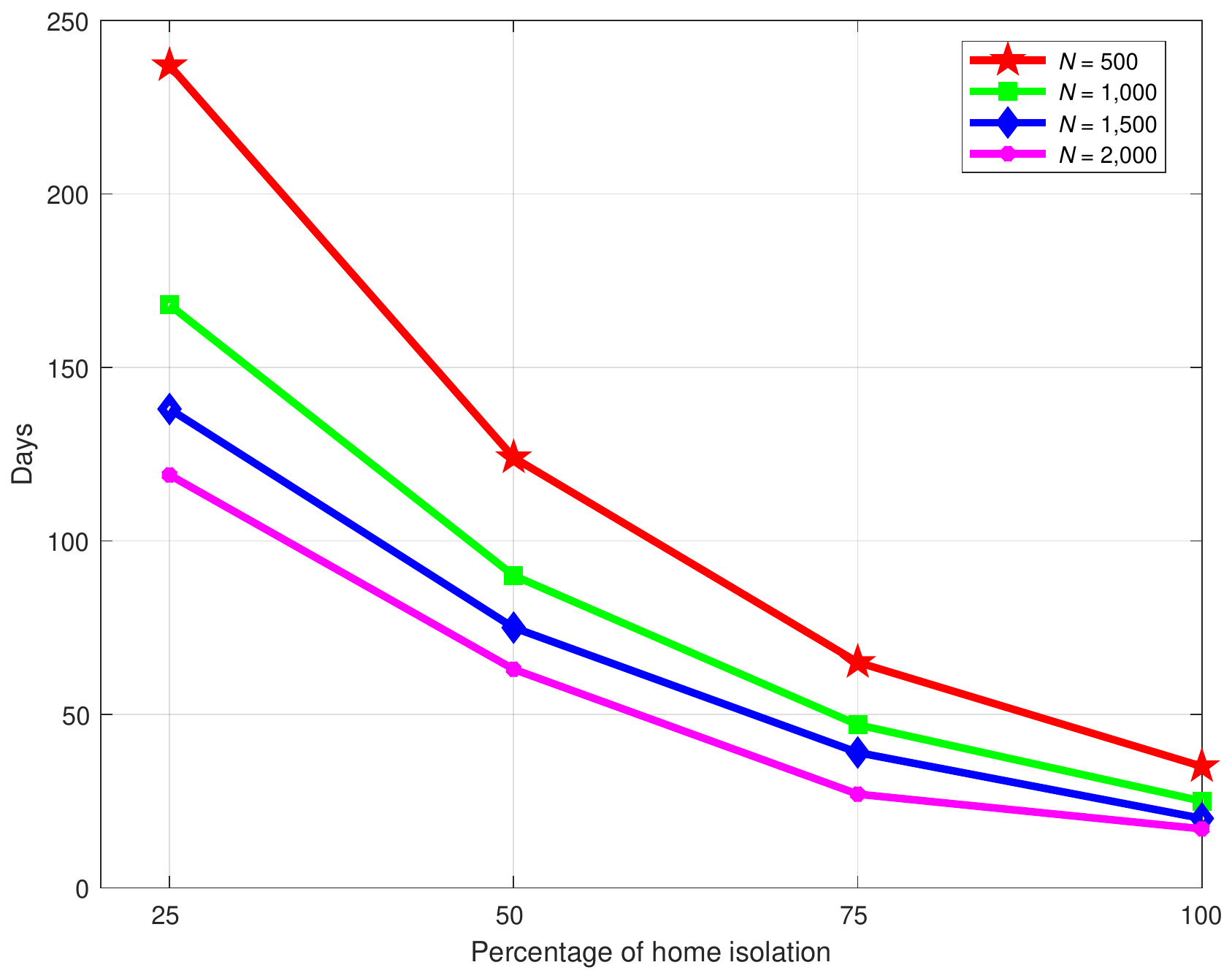}
	\caption{Maximum possible lockdown period with varying number of individuals when \textcolor{black}{$R_0=5\times 10^{23}$}, $\textcolor{black}{\tilde{r}=0.10\times\tilde{U}}$, and using total incentive shown in Figure $\ref{Fig_bar}$.}\label{Fig_day1}
	%$\vspace{-.4cm}$
\end{figure}
\begin{figure}
	\centering
	\includegraphics[width=.5\textwidth]{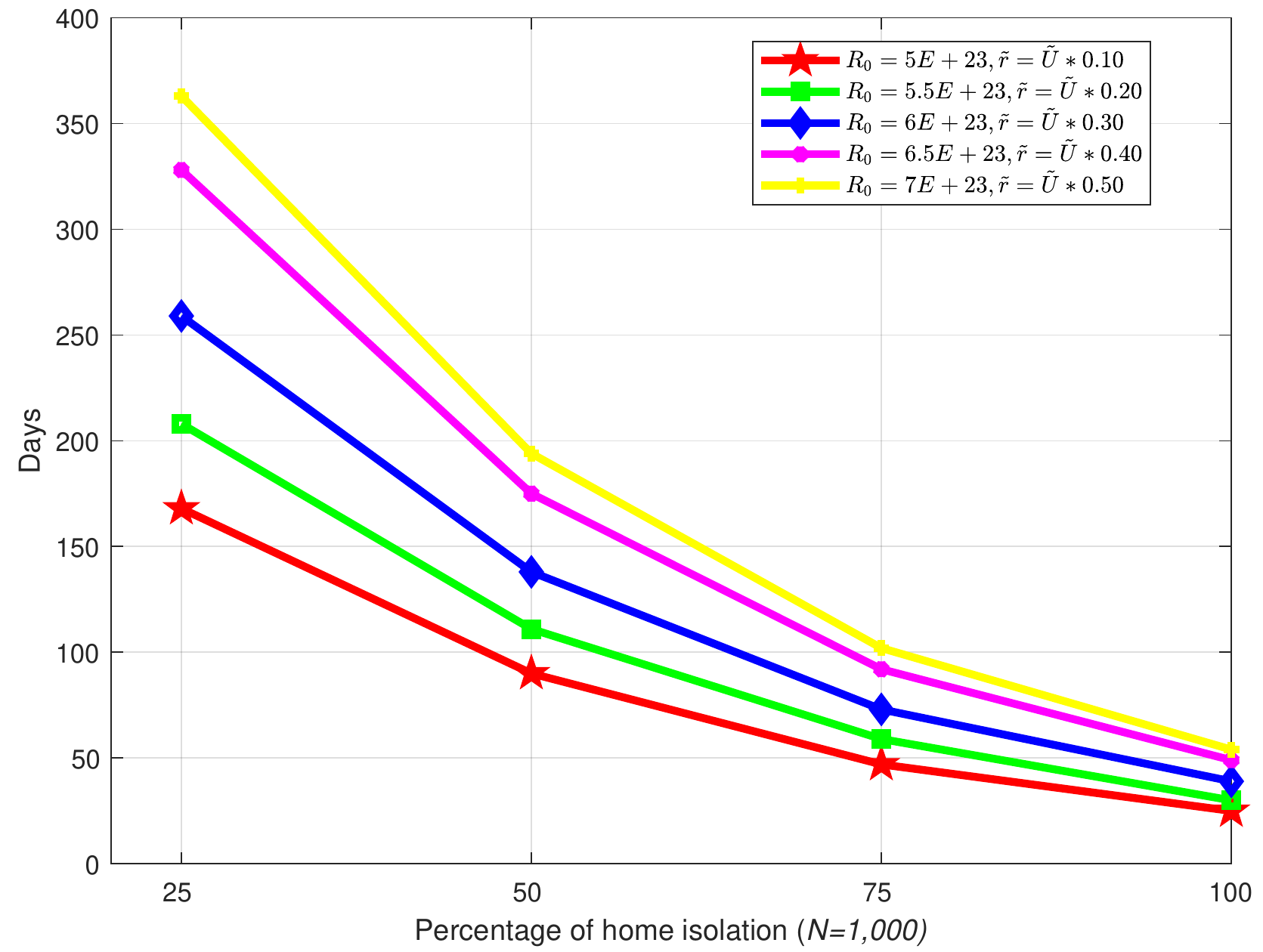}
	\caption{Maximum possible lockdown period with varying $R_0$ and $\tilde{r}$ with total payoff shown in Figure $\ref{Fig_bar}$.}\label{Fig_day2}
	%$\vspace{-.4cm}$
\end{figure}

Figure \ref{Fig_day1} shows the maximum possible lockdown period for a varying number of individuals within a fixed amount of resource $R_0$. The figure reveals that with the increasing percentage of home isolation individuals, the maximum lockdown period significantly decreases for all considered cases. \textcolor{black}{The reason behind this is that the more individuals are in home isolation, the more it is necessary to pay the incentives.} With a fixed amount of resources, a country with less individuals can survive a longer lockdown period. With more percentages of home isolation individuals, the number of \textcolor{black}{lockdown} period is less, and possible of spreading of COVID-19 is also less. Therefore, the governments can consider a trade-off between increasing expenditure as a incentive and lockdown period. For $1,000$ individuals, the maximum possible lockdown period for varying amount of $R_0$ and $\tilde{r}$ is presented in Figure \ref{Fig_day2}. The figure also illustrates that with the increasing percentages of home isolation individuals, the continuity of the lockdown period reduces for every scenarios. However, for a particular percentage of home isolation individuals where total number of individuals are fixed, a country can continue higher lockdown period who has more am amount of resources, $R_0$. Additionally, $\tilde{r}$ also play an important role to continue the lockdown period. 

\section{Discussion and Limitation of the study}
\label{Limitation}
\textcolor{black}{The limitations of the current study are summarized as follows:}
\begin{itemize}  
	%\item First, we analyze the real-world dataset to realize the worldwide severity of COVID-19 epidemic and also show the predicted results for infected and active cases of COVID-19.
	\item \textcolor{black}{The governments are required to formulate and update their lockdown policy over time while mitigating the financial impact of COVID-19. From equations \eqref{Equ_U1_2} - \eqref{Equ_U1_5} we observe that such a policy update requires significant financial planning and prediction of the government revenue over a timespan while considering the socio-economic conditions of the population. In other words, the financial condition of the population has a considerable impact on the success of any distributed (i.e., less government control) or centralized (i.e., strict government control) lockdown policies.}
	
	\item \textcolor{black}{The proposed noncooperative game solution provides an analytical approach to attain the solution that keeps the total deviation from home $\delta$ by an individual from a vast population to the minimum for reducing the infection risk of COVID-19. Besides, such a game solution provides a tractable evaluation for the sustainability of the lockdown policy. However, like the other mathematical models of epidemic diseases (i.e., SI, SIR, SIRS, SEIR), incorporating the probabilities of infection and recovery for COVID-19 in the noncooperative game setting is still an open research question and requires further investigation.}   
\end{itemize}
\textcolor{black}{To this end, in the future, we will further study a stochastic game setting by incorporating epidemic cases as a dynamics of Markovian for capturing the uncertain behavior of pandemic. In particular, the policy for the government-controlled epidemic models will be analyzed by the multi-agent noncooperative reinforcement learning for coping with an unknown epidemic environment. Therefore, the government/institute will capable of taking a proactive policy measurement for enhancing the sustainability of any kind of epidemic.}
\section{Conclusions}
\label{Conc}
In this paper, we have introduced a mathematical model for controlling the outbreak of COVID-19 by augmenting isolation and social distancing features of individuals. We have solved the utility maximization problem by using a noncooperative game. Here, we have proved that staying home (home isolation) is the best strategy of every individual and there is a Nash equilibrium of the game. By applying the proposed model, we have also analyzed the sustainability period of a country with a lockdown policy. Finally, we have performed a detailed numerical analysis of the proposed model to control the outbreak of the COVID-19. In future, we will further study and compare with extended cases such as centralized and different game-theoretic models. In particular, an extensive analysis between the government-controlled spread or individual controlled spread under more diverse epidemic models.

\ifCLASSOPTIONcaptionsoff
  \newpage
\fi

\begin{IEEEbiography}[{\includegraphics[width=1in,height=1.25in,clip,keepaspectratio]{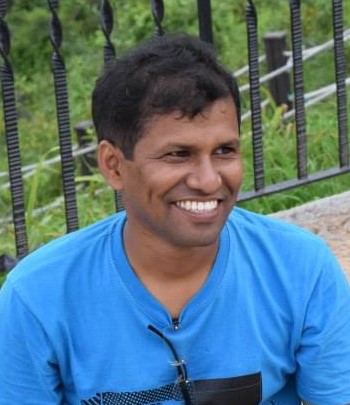}}]{Anupam Kumar Bairagi}(S’17- M’18) received his Ph.D. degree in Computer Engineering from Kyung Hee University, South Korea and B.Sc. and M.Sc. degree in Computer Science and Engineering from Khulna University (KU), Bangladesh. He is an associate professor in Computer Science and Engineering discipline, Khulna University, Bangladesh. His research interests include wireless resource management in 5G and beyond, Healthcare, IIoT, cooperative communication, and game theory. He has authored and coauthored around 40 publications including refereed IEEE/ACM journals, and conference papers. He has served as a technical program committee member in different international conferences. He is a member of  IEEE.
\end{IEEEbiography}

\vskip -1\baselineskip plus -1fil
\begin{IEEEbiography}[{\includegraphics[width=1in,height=1.25in,clip,keepaspectratio]{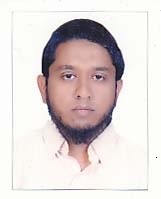}}]{Mehedi Masud} (Senior Member, IEEE) is a Professor in the Department of Computer Science at the Taif University, Taif, KSA. Dr. Mehedi Masud received his Ph.D. in Computer Science from the University of Ottawa, Canada. His research interests include machine learning, distributed algorithms, data security, formal methods, and health analytics. He has authored and coauthored around 70 publications including refereed IEEE/ACM/Springer/Elsevier journals, conference papers, books, and book chapters. He has served as a technical program committee member in different international conferences. He is a recipient of a number of awards including, the Research in Excellence Award from Taif University. He is on the Associate Editorial Board of IEEE Access, International Journal of Knowledge Society Research (IJKSR), and editorial board member of Journal of Software. He also served as a guest editor of ComSIS Journal and Journal of Universal Computer Science (JUCS). Dr. Mehedi is a Senior Member of IEEE, a member of ACM.
\end{IEEEbiography}

\vskip -1\baselineskip plus -1fil
\begin{IEEEbiography}[{\includegraphics[width=1in,height=1.25in,clip,keepaspectratio]{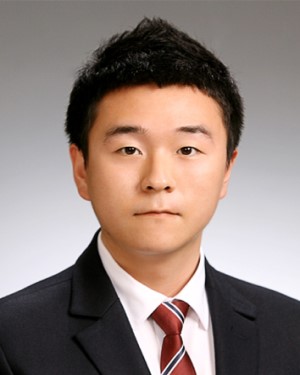}}]{Do Dyeon Kim} received his B.S. degree in Communication Engineering from Jeju National University, in 2014 and received M.S. degree from Kyung Hee University in 2017. He is currently working toward his Ph.D. degree at the Department of Computer Science and Engineering, Kyung Hee University. His research interests include Multi-access Edge Computing, Wireless Network Virtualization.
\end{IEEEbiography}

\vskip -1\baselineskip plus -1fil
\begin{IEEEbiography}[{\includegraphics[width=1in,height=1.25in,clip,keepaspectratio]{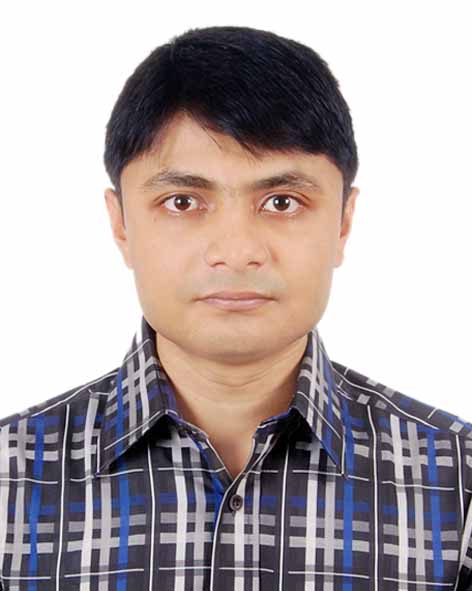}}]{Md.~Shirajum~Munir}
	(S'19) received the B.S. degree in computer science and engineering from Khulna University, Khulna, Bangladesh, in 2010. He is currently pursuing the Ph.D. degree in computer science and engineering at Kyung Hee University, Seoul, South Korea. He served as a Lead Engineer with the Solution	Laboratory, Samsung Research and Development Institute, Dhaka, Bangladesh, from 2010 to 2016. His current research interests include IoT network management, fog computing, mobile edge computing, software-defined networking, smart grid, and machine learning.
\end{IEEEbiography}

\vskip -1\baselineskip plus -1fil
\begin{IEEEbiography}[{\includegraphics[width=1in,height=1.25in,clip,keepaspectratio]{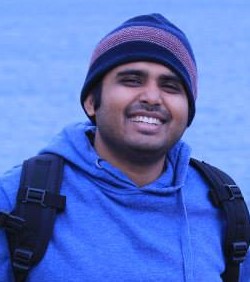}}]{ABDULLAH-AL NAHID} received the B.Sc. degree in electronics and communication engineering from Khulna University, Khulna, Bangladesh, in 2007, the M.Sc. degree in telecommunication engineering from the Institute for the Telecommunication Research (ITR), University of South Australia (UniSA), Australia, in 2014, and the Ph.D. degree from Macquarie University, Sydney Australia, in 2018. His research interests include machine learning, biomedical image processing, data classification, and smart grid.
\end{IEEEbiography}
\vskip -1\baselineskip plus -1fil
\begin{IEEEbiography}[{\includegraphics[width=1in,height=1.25in,clip,keepaspectratio]{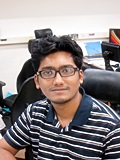}}]{Sarder Fakhrul Abedin} (S'18) received his B.S. degree in Computer Science from Kristianstad University, Kristianstad, Sweden, in 2013. He received his Ph.D. degree in Computer Engineering from Kyung Hee University, South Korea in 2020. Currently, he is serving as a Postdoctoral Researcher at the Department of Computer Science and Engineering, Kyung Hee University, Korea. His research interests include Internet of Things (IoT) network management, Fog computing, and Wireless networking. Dr. Abedin is a Member of the Korean Institute of Information Scientists and Engineers (KIISE).
\end{IEEEbiography}

\vskip -1\baselineskip plus -1fil
\begin{IEEEbiography}[{\includegraphics[width=1in,height=1.25in,clip,keepaspectratio]{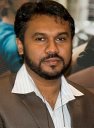}}]{Kazi Masudul Alam} is currently an Associate Professor of Computer Science and Engineering in Khulna University, Bangladesh. He has received his Ph.D. and M.C.S degrees from the University of Ottawa, Canada subsequently in 2017 and 2012. He completed his B.Engg. from Khulna University,  Bangladesh. During his graduate studies, he played the main role in the design process of the Social Internet of Vehicles, Digital Twin Architecture, Haptic E-Book, and Haptic EmoJacket. He has authored and co-authored 25 peer-reviewed international conference and journal articles with the ACM, IEEE, and Springer publishers. He regularly reviews journal articles from the top publishers of IEEE, ACM, Springer as well as reviews project proposals for European funding applications. He has received best paper awards in a few IEEE conferences. He was awarded prestigious academic and research scholarships of NSERC Canada Graduate-Doctoral and Ontario Graduate Scholarship-Masters during his graduate studies.
\end{IEEEbiography}

\vskip -1\baselineskip plus -1fil
\begin{IEEEbiography}[{\includegraphics[width=1in,height=1.25in,clip,keepaspectratio]{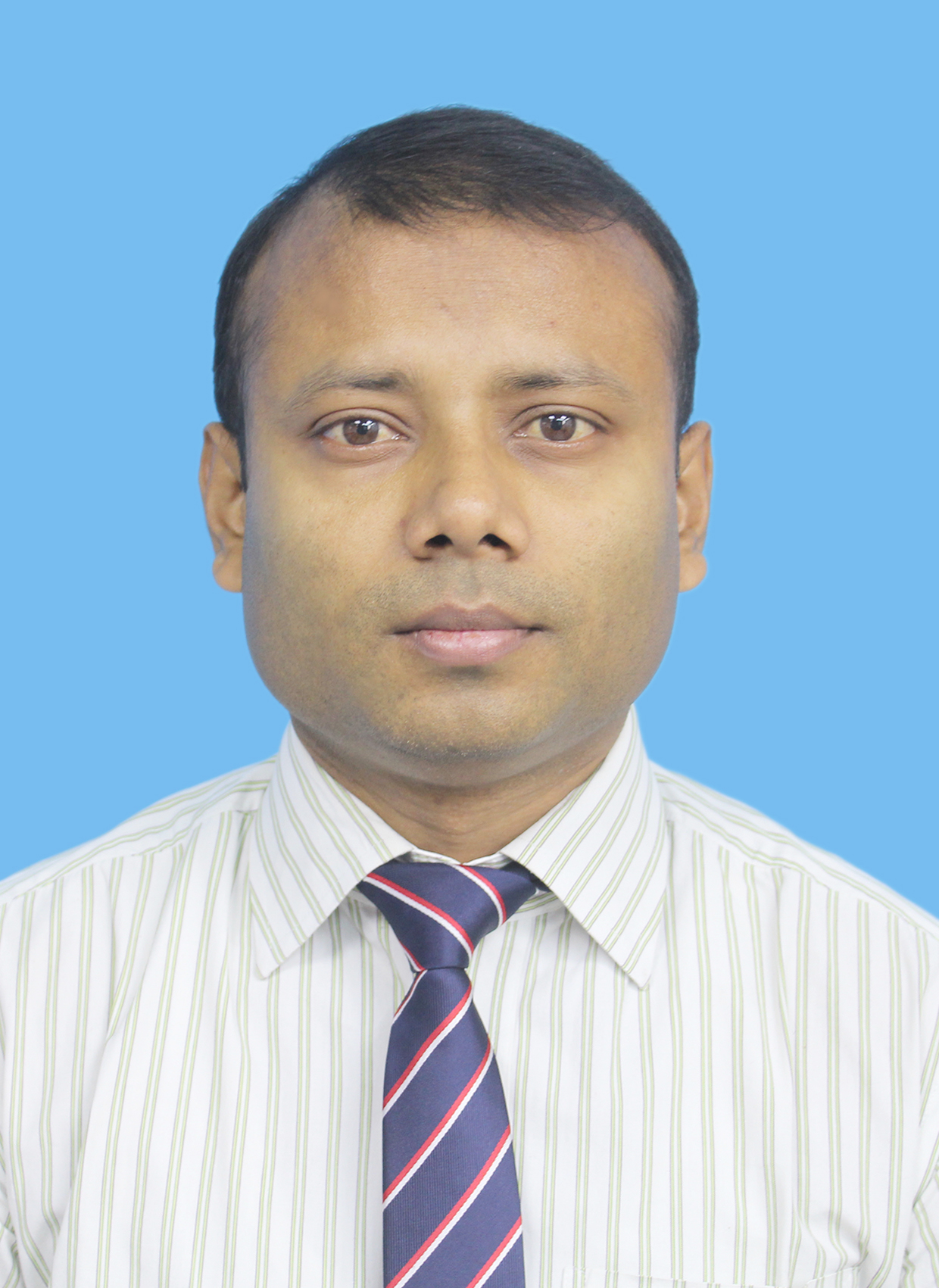}}]{Sujit Biswas} (M’20) received his M.Sc. degree in Computer Engineering from Northwestern Polytechnical University, China in 2015 and Ph.D degree in Computer Science and Technology from Beijing Institute of Technology, China. He is also an Assistant Professor with Computer Science and Engineering department, Faridpur Engineering College, University of Dhaka, Bangladesh. His basic research interest is in IoT, Blockchain, Mobile computing security and privacy, Data driven decision making, etc.
\end{IEEEbiography}

\vskip -1\baselineskip plus -1fil
\begin{IEEEbiography}[{\includegraphics[width=1in,height=1.25in,clip,keepaspectratio]{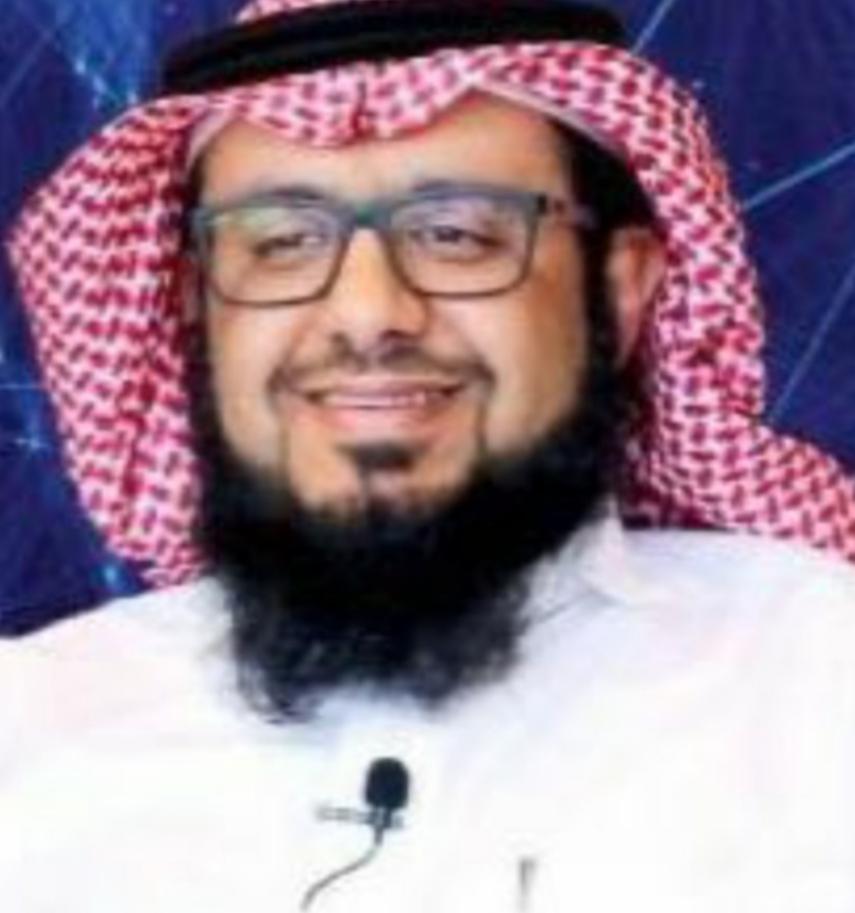}}]{Sultan S Alshamrani} is currently working as an associate professor at Taif University in Saudi Arabia and he is the head of  the department of Information Technology. DR. Sultan got his PhD from the University of Liverpool in UK and a masters degree in Information Technology (Computer Networks) from the University of Sydney in Sydney, Australia. DR. Sultan finished his bachelor's degree in computer science from Taif University in 2007 with General Grade "Excellent" With first honor and an accumulative GPA of (4.85) out of (5.00) where considered the highest GPA in the collage.
\end{IEEEbiography}

\vskip -2.25\baselineskip plus -1fil
\begin{IEEEbiography}[{\includegraphics[width=1in,height=1.25in,clip,keepaspectratio]{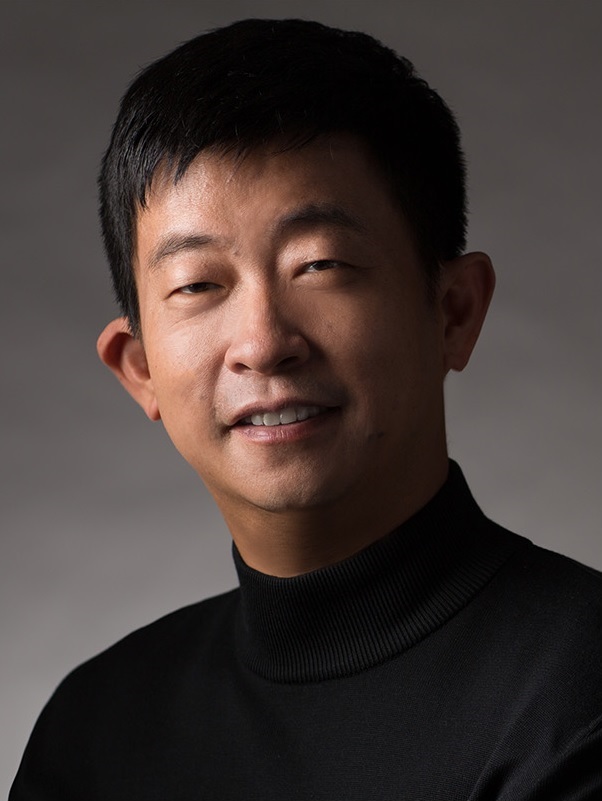}}]{Zhu Han} (S’01–M’04-SM’09-F’14) received the B.S. degree in electronic engineering from Tsinghua University, in 1997, and the M.S. and Ph.D. degrees in electrical and computer engineering from the University of Maryland, College Park, in 1999 and 2003, respectively. From 2000 to 2002, he was an R\&D Engineer of JDSU, Germantown, Maryland. From 2003 to 2006, he was a Research Associate at the University of Maryland. From 2006 to 2008, he was an assistant professor at Boise State University, Idaho. Currently, he is a John and Rebecca Moores Professor in the Electrical and Computer Engineering Department as well as in the Computer Science Department at the University of Houston, Texas. His research interests include wireless resource allocation and management, wireless communications and networking, game theory, big data analysis, security, and smart grid.  Dr. Han received an NSF Career Award in 2010, the Fred W. Ellersick Prize of the IEEE Communication Society in 2011, the EURASIP Best Paper Award for the Journal on Advances in Signal Processing in 2015, IEEE Leonard G. Abraham Prize in the field of Communications Systems (best paper award in IEEE JSAC) in 2016, and several best paper awards in IEEE conferences. Dr. Han was an IEEE Communications Society Distinguished Lecturer from 2015-2018, AAAS fellow since 2019 and ACM distinguished Member since 2019. Dr. Han is 1\% highly cited researcher since 2017 according to Web of Science. Dr. Han is also the winner of 2021 IEEE Kiyo Tomiyasu Award, for outstanding early to mid-career contributions to technologies holding the promise of innovative applications, with the following citation: ``for contributions to game theory and distributed management of autonomous communication networks."
	
\end{IEEEbiography}

\vskip -2\baselineskip plus -1fil
\begin{IEEEbiography}[{\includegraphics[width=1in,height=1.25in,clip,keepaspectratio]{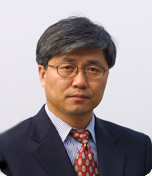}}]{Choong~Seon~Hong}
	(S'95-M'97-SM'11)
	received the B.S. and M.S. degrees in electronic engineering from Kyung Hee University, Seoul, South Korea, in 1983 and 1985, respectively, and the Ph.D. degree from Keio University, Japan, in 1997. In 1988, he joined KT, where he was involved in broadband networks as a Member of Technical Staff. Since 1993, he has been with Keio University. He was with the Telecommunications Network Laboratory, KT, as a Senior Member of Technical Staff and as the Director of the Networking Research Team until 1999. Since 1999, he has been a Professor with the Department of Computer Science and Engineering, Kyung Hee University. His research interests include future Internet, ad hoc networks, network management, and network security. He is a member of the ACM, the IEICE, the IPSJ, the KIISE, the KICS, the KIPS, and the OSIA. Dr. Hong has served as the General Chair, the TPC Chair/Member, or an Organizing Committee Member of international conferences such as NOMS, IM, APNOMS, E2EMON, CCNC, ADSN, ICPP, DIM, WISA, BcN, TINA, SAINT, and ICOIN. He was an Associate Editor of the IEEE TRANSACTIONS ON NETWORK AND SERVICE MANAGEMENT, and the IEEE JOURNAL OF COMMUNICATIONS AND NETWORKS. He currently serves as an Associate Editor of the International Journal of Network Management, and an Associate Technical Editor of the IEEE Communications Magazine.
\end{IEEEbiography}

\end{document}